\documentclass[runningheads]{llncs}

\usepackage{cite}
\usepackage{xcolor}
\usepackage{enumerate}
\usepackage{paralist}
\usepackage{graphicx}
\usepackage{caption}
\usepackage{subcaption}
\usepackage{textcomp}
\usepackage{hyperref}
\usepackage{algorithm}
\usepackage[noend]{algpseudocode}
\usepackage{amsmath, amssymb, amsfonts}
\newcommand{\set}[1]{\{#1\}}
\newcommand{\code}[2]{#1:#2}
\newcommand{\hStatex}[0]{\vspace{4pt}}

\algblockdefx[upon]{When}{EndWhen}%
  [2][Unknown]{\textbf{when for every} #1 \textbf{received} #2}%
  {\textbf{end}}

\algblockdefx[upon]{WhenReceive}{EndWhenReceive}%
  [1][Unknown]{\textbf{when received} #1}%
  {\textbf{end}}

\algblockdefx[upon]{Upon}{EndUpon}%
  [1][Unknown]{\textbf{upon} #1}%
  {\textbf{end upon}}

\makeatletter
\ifthenelse{\equal{\ALG@noend}{t}}%
  {\algtext*{EndWhen}}
  {}%
\makeatother

\makeatletter
\ifthenelse{\equal{\ALG@noend}{t}}%
  {\algtext*{EndWhenReceive}}
  {}%
\makeatother

\makeatletter
\ifthenelse{\equal{\ALG@noend}{t}}%
  {\algtext*{EndUpon}}
  {}%
\makeatother

\newcommand{\cvar}{\mathit{c}}
\newcommand{\D}{D} % data centers
\newcommand{\dvar}{\mathit{d}}
\newcommand{\ivar}{\mathit{i}} % data centers: index i
\renewcommand{\P}{P} % partitions
\newcommand{\pvar}{\mathit{p}}
\newcommand{\jvar}{\mathit{j}} % partition: index j
\newcommand{\svar}{\mathit{s}}
\newcommand{\servers}{\mathcal{S}}

\newcommand{\clock}{{\sf clock}}
\newcommand{\clockvar}{\mathit{clock}}

\newcommand{\partition}{\textsc{partition}}
\newcommand{\datacenter}{\textsc{datacenter}}
\newcommand{\replicas}{\textsc{replicas}}
\newcommand{\replica}[2]{\mathit{r}^{\mathit{#1}}_{\mathit{#2}}}

\newcommand{\store}{{\sf store}}
\newcommand{\storevar}{\mathit{store}}
\newcommand{\stvar}{\mathit{st}}
\newcommand{\Store}{\textsc{Store}}

\newcommand{\Q}{Q} % quorum
\newcommand{\quorum}{{\sf quorum}}

\newcommand{\M}{M} % messages
\newcommand{\mvar}{\mathit{m}}
\newcommand{\sign}[3]{\langle #1 \rangle^{#2}_{#3}}
\newcommand{\gentlerain}{GentleRain}
\newcommand{\byzgentlerain}{Byz-GentleRain}
\newcommand{\byzrcm}{Byz-RCM}
\newcommand{\byzcc}{Byz-CC}
\newcommand{\rel}[1]{\xrightarrow{#1}}
\newcommand{\so}{{\sf so}}
\newcommand{\rf}{{\sf rf}}
\newcommand{\ovar}{\mathit{o}}
\newcommand{\rvar}{\mathit{r}}
\newcommand{\wvar}{\mathit{w}}
\newcommand{\get}{\textsc{get}}
\newcommand{\getreq}{\textsc{get\_req}}
\renewcommand{\put}{\textsc{put}}
\newcommand{\putreq}{\textsc{put\_req}}

\newcommand{\heartbeat}{\textsc{heartbeat}}
\newcommand{\hb}{\textsc{hb}}

\newcommand{\newcgst}{\textsc{new\_cgst}}

\newcommand{\collect}{\textsc{collect}}
\newcommand{\collectack}{\textsc{collect\_ack}}
\newcommand{\propose}{\textsc{propose}}

\newcommand{\prepared}{\textsc{prepared}}
\newcommand{\preparedpred}{{\sf prepared}}
\newcommand{\commit}{\textsc{commit}}

\newcommand{\newview}{\textsc{newview}}
\newcommand{\newleader}{\textsc{newleader}}
\newcommand{\Key}{{\sf Key}}

\newcommand{\kvar}{\mathit{k}}

\newcommand{\Val}{{\sf Val}}
\newcommand{\vvar}{\mathit{v}}
\newcommand{\VVal}{{\sf VVal}}
\newcommand{\vvvar}{\mathit{vv}}

\newcommand{\uvar}{\mathit{u}}
\newcommand{\getack}{\textsc{get\_ack}}
\newcommand{\putack}{\textsc{put\_ack}}
\newcommand{\var}{{\bf var}\;}
\newcommand{\ok}{{ok}}
\newcommand{\pre}{{\bf pre:}\;}
\newcommand{\send}{{\bf send}\;}
\newcommand{\sendto}{\;{\bf to}\;}

\newcommand{\receive}{{\bf receive}\;}
\newcommand{\from}{\;{\bf from}\;}
\newcommand{\wait}{{\bf wait}\;}
\newcommand{\until}{{\bf until}\;}

\newcommand{\client}{{\bf client}\;}

\newcommand{\timepoint}{\sigma}
\newcommand{\dt}{{\sf dt}}
\newcommand{\dtvar}{\mathit{dt}}
\newcommand{\utvar}{\mathit{cl}}
\newcommand{\lst}{{\sf lst}}
\newcommand{\lstvar}{\mathit{lst}}
\newcommand{\gst}{{\sf gst}}
\newcommand{\gstvar}{\mathit{gst}}
\newcommand{\cgst}{{\sf cgst}}
\newcommand{\cgstvar}{\mathit{cgst}}

\newcommand{\certvar}{\mathit{cert}}
\newcommand{\currview}{{\sf curr\_view}}
\newcommand{\preparedview}{{\sf prepared\_view}}
\newcommand{\preparedstore}{{\sf prepared\_store}}
\newcommand{\voted}{{\sf voted}}
\newcommand{\true}{{\sf true}}
\newcommand{\false}{{\sf false}}
\newcommand{\tsvar}{\mathit{ts}}

\newcommand{\valid}{{\sf valid}}
\newcommand{\hash}{{\sf hash}}
\newcommand{\hvar}{\mathit{h}}
\newcommand{\currstore}{{\sf curr\_store}}
\newcommand{\currcgst}{{\sf curr\_cgst}}

\newcommand{\viewvar}{\mathit{view}}

\newcommand{\leader}{{\sf leader}}
\newcommand{\ValidNewLeader}{{\sf ValidNewLeader}}
\newcommand{\C}{C} % certificate

\newcommand{\VV}{{\sf VV}}
\newcommand{\LV}{{\sf LV}}

\newcommand{\replicate}{\textsc{replicate}}

\newcommand{\broadcast}{\textsc{broadcast}}

\newcommand{\bc}{\textsc{bc}}

\newcommand{\cert}{{\sf cert}}

\newcommand{\safecollect}{{\sf safe\_collect}}

\newcommand{\safepropose}{{\sf safe\_propose}}

\newcommand{\inv}{\textsc{Inv}}
\newcommand{\byzrule}{\textsc{Rule}}
\newcommand{\tsof}{\textsl{ts}}
\newcommand{\R}{R}
\newcommand{\W}{W}
\renewcommand{\O}{O}

\begin{document}

\title{Byz-GentleRain: An Efficient Byzantine-tolerant Causal Consistency Protocol}

\titlerunning{Byz-GentleRain}

\author{Kaile Huang\inst{1} \and
Hengfeng Wei\inst{2,3}\thanks{Corresponding Author} \and
Yu Huang\inst{2} \and
Haixiang Li\inst{4} \and
Anqun Pan\inst{4}}

\authorrunning{K. Huang et al.}
\institute{
  Computer Science and Technology, Nanjing University, China
  \email{MG1933024@smail.nju.edu.cn} \\
  \and
  State Key Laboratory for Novel Software Technology,
    Nanjing University, China
  \email{\{hfwei, yuhuang\}@nju.edu.cn} \\
  \and
  Software Institute, Nanjing University, China
  \and
  Tencent Inc., China \\
  \email{\{blueseali, aaronpan\}@tencent.com}}

\maketitle
%%%%%%%%%%%%%%%%%%%%%%
% abstract.tex

\begin{abstract}
  Causal consistency is a widely used weak consistency model
  that allows high availability despite network partitions.
  There are plenty of research prototypes
  and industrial deployments of causally consistent distributed systems.
  However, as far as we know, none of them consider Byzantine faults,
  except Byz-RCM proposed by Tseng et al.
  Byz-RCM achieves causal consistency in the client-server model
  with $3f + 1$ servers where up to $f$ servers may suffer Byzantine faults,
  but assumes that clients are non-Byzantine.
  In this work, we present Byz-Gentlerain, the first
  causal consistency protocol which tolerates up to $f$ Byzantine
  servers among $3f + 1$ servers in each partition
  and any number of Byzantine clients.
  Byz-GentleRain is inspired by the stabilization mechanism
  of GentleRain for causal consistency.
  To prevent causal violations due to Byzantine faults,
  Byz-GentleRain relies on PBFT to reach agreement on a sequence of
  global stable times and updates among servers,
  and only updates with timestamps less than or equal to
  such common global stable times are visible to clients.
  We prove that Byz-GentleRain achieves Byz-CC,
  the causal consistency variant in the presence of Byzantine faults.
  We evaluate Byz-GentleRain on Aliyun.
  The preliminary results show that Byz-GentleRain is efficient
  on typical workloads.
\end{abstract}
% intro.tex

%%%%%%%%%%%%%%%%%%%%%%%%%%%%%%%%%%%%%%%%
\section{Introduction} \label{section:intro}

% {\emph{Weak Consistency.}}
For high availability and low latency
even under network partitions,
distributed systems often partition and replicate data
among multiple nodes~\cite{RDT:POPL14}.
Due to the CAP theorem~\cite{CAP:PODC00}
% and the PACELC tradeoff~\cite{PACELC:Computer12},
many distributed systems choose to sacrifice strong consistency
and to implement weak ones.

% {\emph{Causal Consistency.}}
Causal consistency~\cite{CM:DC95} is one of
the most widely used consistency model in distributed systems.
% In some settings, it has been shown to be
% the strongest model that allows availability
% in the presence of network partitions~\cite{CAC:TR2011, Attiya:PODC2016}.
There are several variants of causal consistency
in the literature~\cite{CM:DC95, CC:PPoPP16, VCC:POPL17, Jiang:SRDS20}.
% such as causal memory~\cite{CM:DC95},
% causal convergence~\cite{PoEC:NOW14}, and
% weak causal consistency~\cite{CC:PPoPP16}.
They all guarantee that
\emph{an update does not become visible until all its causality are visible.}
% \emph{\red{???clients see updates in an order
% that respects the potential causality between them.}}
We informally explain it in the ``Lost-Ring'' example~\cite{Chapar:POPL16}.
Alice first posts that she has lost her ring.
After a while, she posts that she has found it.
Bob sees Alice's two posts, and comments that ``Glad to hear it''.
We say that there is a \emph{read-from} dependency
from Alice's second post to Bob's get operation,
and a \emph{session} dependency
from Bob's get operation to his own comment.
By \emph{transitivity}, Bob's comment causally depend on Alice's second post.
Thus, when Carol, a friend of Alice and Bob, sees Bob's comment,
she should also see Alice's second post.
If she saw only Alice's first post,
she would mistakenly think that Bob is glad to hear that Alice has lost her ring.

% {\emph{Causal Consistency Protocols.}}
There are plenty of research prototypes
and industrial deployments of causally consistent
distributed systems (e.g., COPS~\cite{COPS:SOSP11},
Eiger~\cite{Eiger:NSDI13},
% Bolt-on~\cite{BCC:SIGMOD13},
\gentlerain~\cite{GentleRain:SoCC14},
Cure~\cite{Cure:ICDCS16},
MongoDB~\cite{MongoDB:SIGMOD19}, and
\byzrcm~\cite{ByzRCM:NAC19}).
% UniStore~\cite{UniStore:ATC21}).
\gentlerain{} uses a stabilization mechanism
to make updates visible while respecting causal consistency.
It timestamps all updates with the physical clock value
of the server where they originate.
Each server $\svar$ periodically computes a global stable time $\gst$,
which is a lower bound on the physical clocks of all servers.
This ensures that no updates with timestamps $\le \gst$ will be generated.
Thus, it is safe to make the updates with timestamps $\le \gst$
at $\svar$ visible to clients.
A get operation with dependency time $\dtvar$ issued to $\svar$
will wait until $\gst \ge \dtvar$
and then obtain the latest version before $\gst$.

% {\emph{Byzantine-tolerant Causal Consistency Protocols.}}
However, none of these causal consistency protocols/systems
consider Byzantine faults,
except Byz-RCM (Byzantine Resilient Causal Memory) in~\cite{ByzRCM:NAC19}.
Byz-RCM achieves causal consistency in the client-server model with $3f + 1$ servers
where up to $f$ servers may suffer Byzantine faults,
and any number of clients may crash.
Byz-RCM has also been shown optimal in terms of failure resilience.
However, Byz-RCM did not tolerate Byzantine clients,
and thus it could rely on clients' requests
to identify bogus requests from Byzantine servers~\cite{ByzRCM:NAC19}.

% {\emph{Our \byzgentlerain{} Protocol.}}
In this work, we present \byzgentlerain,
the first Byzantine-tolerant causal consistency protocol
which tolerates up to $f$ Byzantine servers among $3f + 1$ servers in each partition
\emph{and} any number of Byzantine clients.
It uses PBFT~\cite{PBFT:Thesis00} to reach agreement among servers on
a total order of client requests.
The major challenge Byz-GentleRain faces is to
ensure that the agreement is consistent with the causal order.
To this end, Byz-GentleRain should prevent causality violations
caused by Byzantine clients or servers:
Byzantine clients may violate the session order
by fooling some servers that a request happened before another
that was issued earlier.
Byzantine servers may forge causal dependencies
by attaching arbitrary metadata for causality tracking
to the forward messages.
To migrate the potential damages of Byzantine servers,
we let clients assign totally ordered timestamps to updates in \byzgentlerain.
Utilizing the digital signatures mechanism,
Byzantine servers cannot forge causal dependencies.

\begin{figure}[t]
  \centering
  \includegraphics[width = 0.70\textwidth]{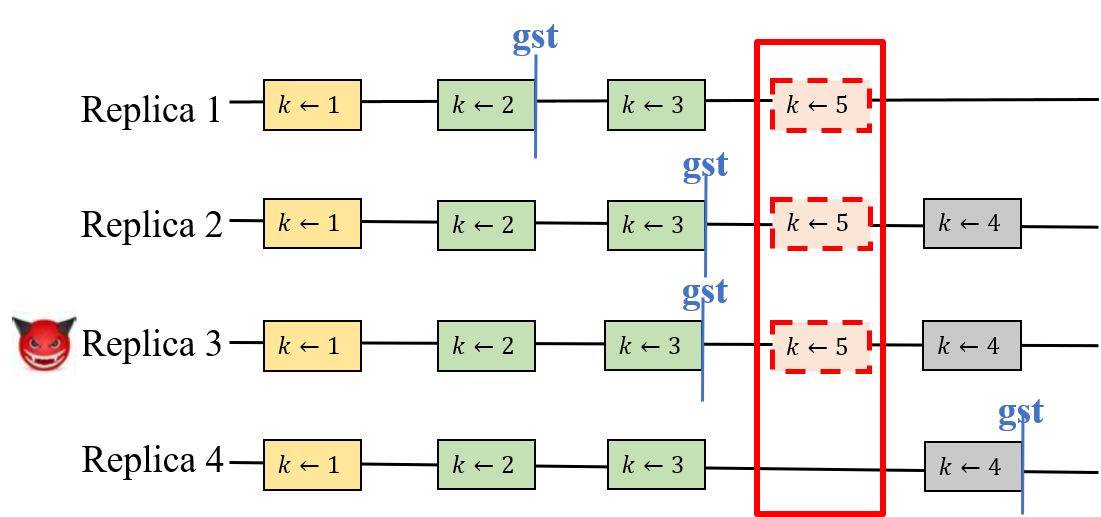}
  \caption{Why the servers in \byzgentlerain{} need to synchronize global stable times.}
  \label{fig:example-gst}
\end{figure}

To preserve causality,
\byzgentlerain{} uses the stabilization mechanism of \gentlerain.
As explained above, the timestamps in \byzgentlerain{} are generated by clients.
However, it is unrealistic to compute a lower bound on physical clock values
of an arbitrary number of clients.
Therefore, each server $\svar$ in \byzgentlerain{} maintains
and periodically computes a global stable time $\gst$
which is a lower bound on physical clock values of the clients it is aware of.
In the following, we argue that simply refusing any updates
with timestamps $\le \gst$ on each server as \gentlerain{} does
may lead to causality violations.
Consider a system of four servers which are replicas
all maintaining a single key $\kvar$,
as shown in Figure~\ref{fig:example-gst}.
Due to asynchrony, these four servers may have different values of $\gst$.
Without loss of generality, we assume that
$\gst_{1} < \gst_{2} = \gst_{3} < \gst_{4}$,
as indicated by vertical lines.
Now suppose that a new update $\uvar: \kvar \gets 5$
with timestamp between $\gst_{3}$ and $\gst_{4}$ arrives,
and we want to install it on $\ge 3$ servers,
using quorum mechanism.
In this scenario, if each server refuses any updates with timestamps
smaller than or equal to its $\gst$,
the update $\uvar$ can only be accepted
by the first 3 servers, indicated by dashed boxes.
Suppose that server 3 is a Byzantine server,
which may expose or hide the update $\uvar$ as it will.
Consequently, later read operations which read from $\ge 3$ servers
may or may not see this update $\uvar$.
That is, the Byzantine server $3$ may cause causality violations.

To cope with this problem,
we need to synchronize the global stable times of servers.
When a server periodically computes its $\gst$,
it checks whether no larger global stable time
has been or is being synchronized.
If so, the server will try to synchronize its $\gst$
among all servers, by running PBFT independently in each partition.
For each partition, the PBFT leader is also responsible for
collecting updates with timestamps $\le \gst$ from $2f + 1$ servers,
and synchronizing them on all servers.
Once successfully synchronized,
a global stable time becomes a \emph{common} global stable time,
denoted $\cgst$, and in each partition
the updates with timestamps $\le \cgst$ on all correct servers are the same.
Therefore, each server can safely refuse any updates with timestamps
smaller than or equal to its $\cgst$.

Still, the classic PBFT is insufficient to guarantee causality,
since a Byzantine leader of each partition
may propose an arbitrary set of updates.
To avoid this, in \byzgentlerain{} the PBFT leader will also
include the sets of updates it collects from $2f+1$ servers
in its \propose{} message.
A server will reject the \propose{} message
if it finds the contents of this message have been manipulated
by checking hash and signatures.

% {\emph{Contributions.}}
Thus, we make the following contributions:
\begin{itemize}
  \item We define Byzantine Causal Consistency (\byzcc),
    which is a causal consistency variant in the presence of Byzantine faults
    (Section~\ref{section:cc}).
  \item We present \byzgentlerain,
    the first Byzantine-tolerant causal consistency protocol.
    It tolerates up to $f$ Byzantine servers among $3f + 1$ ones
    and any number of Byzantine clients
    (Section~\ref{section:protocol}).
    All reads and updates complete in one round-trip.
  \item We evaluate \byzgentlerain{} on Aliyun.
    The preliminary results show that \byzgentlerain{} is efficient
    on typical workloads (Section~\ref{section:evaluation}).
\end{itemize}

% The rest of the paper is organized as follows.
Section~\ref{section:model} describes the system model and failure model.
% Section~\ref{section:cc} defines the correctness condition \byzcc.
% Section~\ref{section:protocol} presents our \byzgentlerain{} protocol.
% Section~\ref{section:evaluation} evaluates the performance of \byzgentlerain.
Section~\ref{section:related-work} discusses related work.
Section~\ref{section:conclusion} concludes the paper.
The proofs can be found in Appendix~\ref{section:proof}.
%%%%%%%%%%%%%%%%%%%%%%%%%%%%%%%%%%%%%%%%
% model.tex

%%%%%%%%%%%%%%%%%%%%%%%%%%%%%%%%%%%%%%%%
\section{Model} \label{section:model}

% notations.tex

\begin{table}[t]
  \centering
  \caption{Notations.}
  \label{table:notations}
  \begin{tabular}{l|l}
    \hline
    \textbf{Notations} & \textbf{Meaning}\\
    \hline
    % $\C$ & number of clients \\
    $\D$ & number of data centers \\
    $\P$ & number of partitions \\
    % $\S$ & number of servers ($\S = \D \times \P$) \\
    $\replica{\pvar}{\dvar}$ & the replica of partition $\pvar$ in data center $\dvar$ \\
    \hline
    $\partition(\kvar)$ & the partition that holds key $\kvar$ \\
    $\replicas(\pvar)$ & the set of replicas of partition $\pvar$ \\
    $\datacenter(\dvar)$ & the set of servers in data center $\dvar$ \\
    % $\replicasof(\kvar)$ & the set of replicas that hold key $\kvar$ \\
    $\servers$ & the set of all $\D \times \P$ servers in the key-value store \\
    \hline
    $\clock_{\cvar}$         & clock at client $\cvar$ \\
    $\dt_{\cvar}$  & dependency time at client $\cvar$ \\
    $\cgst_{\cvar}$ & the maximum common global stable time known by client $\cvar$ \\
    \hline
    $\clock^{\pvar}_{\dvar}$ & clock at replica $\replica{\pvar}{\dvar}$ \\
    $\lst^{\pvar}_{\dvar}$   & local stable time at replica $\replica{\pvar}{\dvar}$ \\
    $\gst^{\pvar}_{\dvar}$   & global stable time at replica $\replica{\pvar}{\dvar}$ \\
    $\cgst^{\pvar}_{\dvar}$  & common global stable time at replica $\replica{\pvar}{\dvar}$ \\
    $\currcgst^{\pvar}_{\dvar}$ & temporary common global stable time
                                  at replica $\replica{\pvar}{\dvar}$ during PBFT \\
    \hline
    $\VV^{\pvar}_{\dvar}$ & version vector at replica $\replica{\pvar}{\dvar}$ \\
    $\LV^{\pvar}_{\dvar}$ & local stable time vector at replica $\replica{\pvar}{\dvar}$ \\
    % $\GV^{\pvar}_{\dvar}$ & global stable time vector at replica $\replica{\pvar}{\dvar}$ \\
    \hline
    $\Key$  & the set of keys, ranged over by $\kvar$ \\
    $\Val$  & the set of values, ranged over by $\vvar$ \\
    $\VVal$ & the set of versioned values, ranged over by $\vvvar$ \\
    % $\kvar \in \Key$ & key \\
    % $\vvar \in \Val$ & value \\
    % $\vvvar \in \VVal$ & versioned value \\
    $\store^{\pvar}_{\dvar} \subseteq \VVal$
      & store maintained at replica $\replica{\pvar}{\dvar}$ \\
    $\Store$ & the union of stores at all replicas, that is,
      $\Store \triangleq \bigcup_{1 \le i \le \D, 1 \le j \le \P}
      \store^{\jvar}_{\ivar}$  \\
    \hline
  \end{tabular}
\end{table}

We adopt the client/server architecture~\cite{Liskov:ICDCS16, ByzRCM:NAC19},
in which each client or server has its unique id.
Table~\ref{table:notations} summarizes the notations used in this paper.
%%%%%%%%%%%%%%%%%%%%%%%%%%%%%%
\subsection{System Model} \label{ss:model}

We consider a distributed multi-version key-value store,
which maintains keys in the set $\Key$ (ranged over by $\kvar$)
with values in the set $\Val$ (ranged over by $\vvar$).
Each value is associated with a unique version,
consisting of the timestamp of the update which creates this version
and the id of the client which issues this update.
We denote by $\VVal$ (ranged over by $\vvvar$) the set of versioned values.

The distributed key-value store runs at $\D$ data centers,
each of which has a full copy of data.
In each data center, the full data is sharded into $\P$ partitions.
For a key $\kvar \in \Key$,
we use $\partition(\kvar)$ to denote the partition that holds $\kvar$.
Each partition is replicated across $D$ data centers.
For a partition $\pvar$,
we use $\replicas(\pvar)$ to denote the set of replicas of $\pvar$.
For a data center $\dvar$,
we use $\datacenter(\dvar)$ to denote the set of servers in $\dvar$.
We denote by $\replica{\pvar}{\dvar}$
the replica of partition $\pvar$ at data center $\dvar$.
% We also use $\replicasof(\kvar)$ to denote the set of replicas that hold key $\kvar$.
% That is, $\replicasof(\kvar) \triangleq \replicas(\partition(\kvar))$.
We denote by $\servers$ the set of all $\D \times \P$ servers in the key-value store.
For each individual partition $\pvar$,
we call a set $\Q$ of $2f + 1$ replicas in $\replicas(\pvar)$ a quorum
and denote it by $\quorum(\Q)$.

For convenience,
we model the key-value store at replica $\replica{\pvar}{\dvar}$,
denoted $\store^{\pvar}_{\dvar}$, by a set of (unique) versioned values.
That is, $\store^{\pvar}_{\dvar} \subseteq \VVal$.
% In implementation, a mapping from $\Key$ to a \red{subset???} of $\VVal$.
% For key $k$, $\store^{\pvar}_{\dvar}[k]$
% denotes the set of updates on $k$ at replica $\replica{\pvar}{\dvar}$.
We denote by $\Store$ the union of stores at all replicas.
That is, $\Store \triangleq \bigcup_{1 \le i \le \D, 1 \le j \le \P}
\store^{\jvar}_{\ivar}$.
The distributed key-value store offers two operations to clients:
\begin{itemize}
	\item $\get(\kvar)$.
    A get operation which returns the value of some version of key $\kvar$.
	\item $\put(\kvar, \vvar)$.
    A put operation which updates key $\kvar$ with value $\vvar$.
    This creates a new version of $\kvar$.
\end{itemize}

We assume that each client or server is equipped with a physical clock,
which is monotonically increasing.
Clocks at different clients are loosely synchronized by a protocol
such as NTP~\footnote{NTP: The Network Time Protocol. \url{http://www.ntp.org/}.}.
The correctness of \byzgentlerain{} does not
depend on the precision of clock synchronization,
but large clock drifts may negatively impact its performance.
%%%%%%%%%%%%%%%%%%%%%%%%%%%%%%
\subsection{Failure Model} \label{ss:failure-model}

Clients and servers are either correct or faulty.
Correct clients and servers obey their protocols,
while faulty ones may exhibit Byzantine behaviors~\cite{PBFT:Thesis00},
by deviating arbitrarily from their protocols.

We assume asynchronous point-to-point communication channels
among clients and servers.
Messages may be delayed, duplicated, corrupted, or delivered out of order.
We do not assume known bounds on message delays.
The communication network is fully connected.
We require that if the two ends of a channel are both correct
and the sender keeps retransmitting a message,
then the message can eventually be delivered.

We also assume the channels are authenticated.
Clients and servers can sign messages using digital signatures when needed.
A message $\mvar$ signed by a client $\cvar$ or a replica $\replica{\pvar}{\dvar}$
is denoted by $\sign{\mvar}{}{\cvar}$ or $\sign{\mvar}{\pvar}{\dvar}$, respectively.
We denote by $\valid(\mvar)$ that $m$ is valid in signatures.
We also use a cryptographic hash function $\hash()$,
which is assumed to be collision-resistant:
the probability of an adversary producing inputs $\mvar$ and $\mvar'$
such that $\hash(\mvar) = \hash(\mvar')$
is negligible~\cite{PBFT:Thesis00,ByzLive:DISC20}.
% With the hashing and digital signature mechanism,
% we \red{can safely assume that} a byzantine node cannot pretend to be someone else
% or maliciously manipulate forwarded messages.
%%%%%%%%%%%%%%%%%%%%%%%%%%%%%%%%%%%%%%%%
% cc.tex

%%%%%%%%%%%%%%%%%%%%%%%%%%%%%%%%%%%%%%%%
\section{Byzantine Causal Consistency} \label{section:cc}

Causal consistency variants in the literature~\cite{
  CM:DC95,CC:PPoPP16,VCC:POPL17,Jiang:SRDS20}
are defined based on the \emph{happens-before} relation
among events~\cite{Lamport:CACM79}.
However, they are not applicable to systems that allows Byzantine nodes,
particularly Byzantine clients.
We now adapt the happens-before relation in Byzantine-tolerant systems,
and define Byzantine Causal Consistency (\byzcc) as follows.
For two events $e$ and $f$,
we say that $e$ happens before $f$, denoted $e \leadsto f$,
if and only if one of the following three rules holds:
\begin{itemize}
  \item \emph{Session-order}.
    Events $e$ and $f$ are two operation requests
    issued by the same \emph{correct} client,
    and $e$ is issued before $f$.
    We denote it by $e \rel{\so} f$.
    We do \emph{not} require session order among operations
    issued by Byzantine clients.
  \item \emph{Read-from relation}.
    Event $e$ is a \put{} request issued by some client
    and $f$ is a \get{} request issued by a \emph{correct} client,
    and $f$ reads the value updated by $e$.
    We denote it by $e \rel{\rf} f$.
    Since a \get{} of Byzantine clients may return an arbitrary value,
    we do \emph{not} require read-from relation induced by it.
  \item \emph{Transitivity}.
    There is another operation request $g$
    such that $e \leadsto g$ and $g \leadsto f$.
\end{itemize}
If $e \leadsto f$, we also say that $f$ causally depends on $e$
and $e$ is a causal dependency of $f$.
A version $\vvvar$ of a key $\kvar$ causally depends on
version $\vvvar'$ of key $\kvar'$,
if the update of $\vvvar$ causally depends on that of $\vvvar'$.
A key-value store satisfies \byzcc{} if,
when a certain version of a key is visible to a client,
then so are all of its causal dependencies.
%%%%%%%%%%%%%%%%%%%%%%%%%%%%%%%%%%%%%%%%
% protocol.tex

%%%%%%%%%%%%%%%%%%%%%%%%%%%%%%%%%%%%%%%%
\section{The Byz-GentleRain Protocol} \label{section:protocol}

As discussed in Section~\ref{section:intro},
it is the clients in \byzgentlerain{} that are responsible
for generating totally ordered timestamps for updates.
Specifically, when a client issues an update,
it assigns to the update a timestamp
consisting of its current clock and identifier.

As in \gentlerain, we also distinguish between the updates
that have been received by a server
and those that have been made visible to clients.
\byzgentlerain{} guarantees that
an update can be made visible to clients
only if so are all its causal dependencies.
The pseudocode in Algorithms~\ref{alg:client}--\ref{alg:metadata}
dealing with Byzantine faults is underlined.
%%%%%%%%%%%%%%%%%%%%%%%%%%%%%%
% key-design.tex

%%%%%%%%%%%%%%%%%%%%%%%%%%%%%%
\subsection{Key Designs} \label{ss:design}

In \byzgentlerain{}, both clients and servers
maintain a \emph{common global stable time} $\cgst$.
We denote the $\cgst$ at client $\cvar$ by $\cgst_{\cvar}$
and that at replica $\replica{\pvar}{\dvar}$ by
$\cgst^{\pvar}_{\dvar}$.
We maintain the following invariants
that are key to the correctness of \byzgentlerain:
\begin{enumerate}[\inv~(I):]
  % \item \label{inv:cgst-c-cgst-replica}
  %   \red{liveness???} At any time, for any correct client $\cvar$,
  %   $\cgst_{\cvar}$ is a lower bound on the set of $\cgstvar$'s
  %   at correct replicas.
  \item \label{inv:cgst-c-put}
    Consider $\cgst_{\cvar}$ at any time $\timepoint$.
    All updates issued by correct client $\cvar$ after time $\timepoint$
    have a timestamp $> \cgst_{\cvar}$.
  %\item \label{inv:cgst-c-get}
  %  Consider $\cgst_{\cvar}$ at any time $\timepoint$.
  %  All the updates with timestamps $\le \cgst_{\cvar}$
  %  in $\Store$ are visible to the read operations
  %  issued by correct client $\cvar$ after time $\timepoint$.
  \item \label{inv:cgst-replica}
    Consider $\cgst^{\pvar}_{\dvar}$ at any time $\timepoint$.
    No updates with timestamps $\le \cgst^{\pvar}_{\dvar}$
    will be successfully executed at $> f$ correct replicas
    in $\replicas(\pvar)$ after time $\timepoint$.
  \item \label{inv:cgst-updates}
    Consider a $\cgstvar$ value.
    For any two correct replicas $\replica{\pvar}{\dvar}$
    and $\replica{\pvar}{\ivar}$ (where $i \neq d$) of partition $\pvar$,
    if $\cgst^{\pvar}_{\dvar} \ge \cgstvar$
    and $\cgst^{\pvar}_{\ivar} \ge \cgstvar$,
    then the updates with timestamps $\le \cgstvar$
    in $\store^{\pvar}_{\dvar}$ and $\store^{\pvar}_{\ivar}$
    are the same.
\end{enumerate}

\byzgentlerain{} further enforces the following rules
for reads and updates:
\begin{enumerate}[\byzrule~(I):]
  \item \label{rule:get}
    For a correct replica $\replica{\pvar}{\dvar}$,
    any updates with timestamps $> \cgst^{\pvar}_{\dvar}$
    in $\store^{\pvar}_{\dvar}$ are invisible to any clients.
  \item \label{rule:put}
    Any correct replica $\replica{\pvar}{\dvar}$
    will reject any updates with timestamps $\le \cgst^{\pvar}_{\dvar}$.
  \item \label{rule:get-ts}
    For a read operation with timestamp $\tsvar$
    issued by client $\cvar$,
    any correct replica $\replica{\pvar}{\dvar}$
    that receives this operation
    must wait until $\cgst^{\pvar}_{\dvar} \ge \tsvar$
    before it returns a value to client $\cvar$.
\end{enumerate}

In the following sections,
we explain how these invariants and rules are implemented
and why they are important to the correctness.
% For now, we state several lemmas and the main theorem.
% The proofs can be found in Appendix~\ref{section:proof}.
%
% \begin{lemma} \label{lemma:local-causality}
%   For any read operation issued by client $\cvar$,
%   if it reads value $\vvar$ from a correct replica $\replica{\pvar}{\dvar}$,
%   then all the causal dependencies of $v$ at $\replica{\pvar}{\dvar}$
%   are visible to $\cvar$.
% \end{lemma}
%%%%%%%%%%%%%%%%%%%%%%%%%%%%%%
% client-op.tex

%%%%%%%%%%%%%%%%%%%%%%%%%%%%%%
\subsection{Client Operations} \label{ss:client}

%%%%%%%%%%%%%%%%%%%%
% client.tex

\begin{algorithm}[t]
  \caption{Operations at client $\cvar$}
  \label{alg:client}
  \begin{algorithmic}[1]
    \Procedure{\get}{$\kvar$}
      \label{line:procedure-get}
      \State \underline{\var $\tsvar \gets \max\set{\dt, \cgst}$}
        \label{line:get-ts}
      \State \var $\pvar \gets \partition(\kvar)$
        \label{line:get-partition}
      \State $\send \Call{\getreq}{\kvar, \tsvar} \sendto \Call{\replicas}{\pvar}$
        \label{line:get-send-getreq}
l
      \hStatex
      \State \underline{\wait\receive
        $\set{\sign{\Call{\getack}{\vvar_{i}, \cgstvar_{i}}}{\pvar}{i}
          \mid \replica{\pvar}{\ivar} \in Q} = \M$ {\bf for a quorum} $Q$}
        \label{line:get-wait-receive-getack}
      \State \underline{$\cgst \gets \max\set{\cgst,
        \min\limits_{\replica{\pvar}{\ivar} \in Q} \cgstvar_{i}}$}
        \label{line:get-cgst}
      \State \underline{$\vvar \gets$ the majority $\vvar_{i}$ in $\M$}
        \label{line:get-val}

      \State \Return $\vvar$
        \label{line:get-return}
    \EndProcedure

    \Statex
    \Procedure{\put}{$\kvar, \vvar$}
      \label{line:procedure-put}
      \State \var $p \gets \partition(\kvar)$
        \label{line:put-partition}
      \State \wait $\clock > \cgst$
        \label{line:put-wait-clock}
      \State $\send \sign{\Call{\putreq}{
          \sign{\kvar, \vvar, \clock, \cvar}{}{\cvar}}}{}{\cvar}
        \sendto \Call{\replicas}{\pvar}$
        \label{line:put-send-putreq}

      \hStatex
      \State \underline{\wait\receive $\set{\sign{\Call{\putack}{\cgstvar_{i}}}{\pvar}{i}
        \mid \replica{\pvar}{\ivar} \in Q} = \M$ {\bf for a quorum $Q$}}
        \label{line:put-wait-receive-putack}
      \State \underline{$\cgst \gets \max\set{\cgst,
        \min\limits_{\replica{\pvar}{\ivar} \in Q} \cgstvar_{i}}$}
        \label{line:put-cgst}

      \State $\dt \gets \clock$
        \label{line:put-dt}
      \State \Return \ok
        \label{line:put-return}
    \EndProcedure
  \end{algorithmic}
\end{algorithm}
%%%%%%%%%%%%%%%%%%%%

To capture the session order,
each client maintains a dependency time $\dt$,
which is the clock value of its last put operation.
When a client $\cvar$ issues a $\get$ operation on key $\kvar$,
it first takes as $\tsvar$ the minimum of its dependency time $\dt_{\cvar}$
and common global stable time $\cgst_{\cvar}$
(line~\code{\ref{alg:client}}{\ref{line:get-ts}}).
Then it sends a $\getreq$ request with $\tsvar$ to $\replicas(\pvar)$
(line~\code{\ref{alg:client}}{\ref{line:get-send-getreq}}),
where $\pvar$ is the partition holding $\kvar$
(line~\code{\ref{alg:client}}{\ref{line:get-partition}}).
Next, the client waits to receive a set $\M$ of
$\getack$ responses from a quorum $Q$ of $\replicas(\pvar)$
(line~\code{\ref{alg:client}}{\ref{line:get-wait-receive-getack}}).
For each replica $\replica{\pvar}{\ivar} \in Q$,
the $\getack$ response from $\replica{\pvar}{\ivar}$
carries a value $\vvar_{i}$ of key $\kvar$
and its common global stable time $\cgstvar_{i}$
when $\replica{\pvar}{\ivar}$ computes the value $\vvar_{i}$ to return
(line~\code{\ref{alg:replica}}{\ref{line:getreq-value}},
discussed in Section~\ref{ss:replica}).
The client takes the minimum $\cgstvar_{i}$ for $\replica{\pvar}{\ivar} \in Q$,
and uses it to update $\cgst_{\cvar}$ if the latter is smaller
(line~\code{\ref{alg:client}}{\ref{line:get-cgst}}).

Since there are at most $f$ Byzantine replicas in a partition,
at least $f + 1$ \getack{} responses are from correct replicas.
By \byzrule~(\ref{rule:get-ts}) and \inv~(\ref{inv:cgst-updates}),
these responses from correct replicas
contain the same value, denoted $\vvar$.
Hence, $\vvar$ is the majority $\vvar_{i}$ in $\M$
(line~\code{\ref{alg:client}}{\ref{line:get-val}}).
Finally, the client returns $\vvar$
(line~\code{\ref{alg:client}}{\ref{line:get-return}}).

When a client $\cvar$ issues a \put{} operation on key $\kvar$ with value $\vvar$,
it sends a \putreq{} request carrying its $\clock_{\cvar}$
and id $\cvar$ to $\replicas(\pvar)$
(line~\code{\ref{alg:client}}{\ref{line:put-send-putreq}}),
where $\pvar$ is the partition holding $\kvar$
(line~\code{\ref{alg:client}}{\ref{line:put-partition}}).
Next the client waits to receive a set $\M$ of \putack{} responses
from a quorum $Q$ of $\replicas(\pvar)$
(line~\code{\ref{alg:client}}{\ref{line:put-wait-receive-putack}}).
For each replica $\replica{\pvar}{\ivar} \in Q$,
the \putack{} response from $\replica{\pvar}{\ivar}$
carries its common global stable time $\cgstvar_{i}$.
Then, the client takes the minimum $\cgstvar_{i}$
for $\replica{\pvar}{\ivar} \in Q$,
and uses it to update $\cgst$ if the latter is smaller
(line~\code{\ref{alg:client}}{\ref{line:put-cgst}}).
Finally, $\dt_{\cvar}$ is set to the current $\clock_{\cvar}$
(line~\code{\ref{alg:client}}{\ref{line:put-dt}}).
%%%%%%%%%%%%%%%%%%%%%%%%%%%%%%
% replica-op.tex

%%%%%%%%%%%%%%%%%%%%%%%%%%%%%%
\subsection{Operation Executions at Replicas} \label{ss:replica}

%%%%%%%%%%%%%%%%%%%%
% replica.tex

\begin{algorithm}[t]
  \caption{Operation execution at $\replica{\pvar}{\dvar}$}
  \label{alg:replica}
  \begin{algorithmic}[1]
    \WhenReceive[$\Call{\getreq}{\kvar, \tsvar} \from \client \cvar$]
      \label{line:procedure-getreq}
      \State \underline{\wait\until $\cgst \geq \tsvar$}
        \label{line:getreq-wait-until}

      \hStatex
      \State $\vvar \gets \text{the value of key $k$ with the largest timestamp}
        \le \tsvar \text{ in } \store$
        \label{line:getreq-value}
      \State $\send \sign{\Call{\getack}{\vvar, \cgst}}{\pvar}{\dvar}
        \sendto \client \cvar$
        \label{line:getreq-send-getack}
    \EndWhenReceive

    \Statex
    \WhenReceive[$\sign{\Call{\putreq}{\sign{\kvar, \vvar, \utvar, \cvar}{}{\cvar}}}{}{\cvar} \from \client \cvar$]
      \label{line:procedure-putreq}
	  \State \underline{\pre $\utvar \ge \lst$}
      \label{line:putreq-pre}

    \hStatex
    \State \wait $clock \ge \utvar$
    \State \var $\vvvar \gets \sign{\kvar, \vvar, \utvar, \cvar}{}{\cvar}$
      \label{line:putreq-vv}
	  \State $\store \gets \store \cup \set{\vvvar}$
      \label{line:putreq-store}
    \State $\send \sign{\Call{\putack}{\cgst}}{\pvar}{\dvar} \sendto \client \cvar$
      \label{line:putreq-send-putack}

    \hStatex
	  \State $\send \sign{\Call{\replicate}{\vvvar}}{\pvar}{\dvar}
	    \sendto \replicas(\pvar) \setminus \set{\replica{\pvar}{\dvar}}$
		\label{line:putreq-send-replicate}
    \EndWhenReceive
  \end{algorithmic}
\end{algorithm}
%%%%%%%%%%%%%%%%%%%%

When a replica $\replica{\pvar}{\dvar}$
receives a $\getreq(\kvar, \tsvar)$ request from some client $\cvar$,
it first waits until $\cgst \ge \tsvar$
(line~\code{\ref{alg:replica}}{\ref{line:getreq-wait-until}})
where $\tsvar \triangleq \max\set{\dt_{\cvar}, \cgst_{\cvar}}$
(line~\code{\ref{alg:client}}{\ref{line:get-ts}}).
This implements \byzrule~\ref{rule:get-ts},
and is used to ensure the session guarantee on client $\cvar$
and eventual visibility of updates to $\cvar$.
Then the replica obtains the value $\vvar$ of key $\kvar$
in $\store$ which has the largest timestamp before $\tsvar$,
breaking ties with client ids
(line~\code{\ref{alg:replica}}{\ref{line:getreq-value}}).
Finally, it sends a signed $\getack$ response,
along with the value $\vvar$ and its current $\cgst$, to client $\cvar$
(line~\code{\ref{alg:replica}}{\ref{line:getreq-send-getack}}).

When a replica $\replica{\pvar}{\dvar}$
receives a $\putreq(\kvar, \vvar, \utvar, \cvar)$ request from client $\cvar$,
it first checks the precondition $\utvar \ge \currcgst^{\pvar}_{\dvar}$
(line~\code{\ref{alg:replica}}{\ref{line:putreq-pre}}).
This enforces \byzrule~\ref{rule:put},
and prevents fabricated updates with timestamps $\le \cgst$
at $\store^{\pvar}_{\dvar}$ from now on.
If the precondition holds,
the replica adds the new versioned version
$\vvvar \triangleq \sign{\kvar, \vvar, \utvar, \cvar}{}{\cvar}$ signed by $\cvar$ 
(line~\code{\ref{alg:replica}}{\ref{line:putreq-vv}})
to $\store^{\pvar}_{\dvar}$
(line~\code{\ref{alg:replica}}{\ref{line:putreq-store}}).
Then, the replica sends a signed $\putack$ response,
with its $\cgst^{\pvar}_{\dvar}$, to client $\cvar$
(line~\code{\ref{alg:replica}}{\ref{line:putreq-send-putack}}).
Finally, it broadcasts a signed $\replicate$ message with $\vvvar$
to other replicas in partition $\pvar$
(line~\code{\ref{alg:replica}}{\ref{line:putreq-send-replicate}}).
%%%%%%%%%%%%%%%%%%%%%%%%%%%%%%
% metadata-op.tex

%%%%%%%%%%%%%%%%%%%%%%%%%%%%%%
\subsection{Metadata} \label{ss:metadata}

%%%%%%%%%%%%%%%%%%%%
\subsubsection{Replica States} \label{sss:replica-states}

Each replica $\replica{\pvar}{\dvar}$ maintains
a \emph{version vector} $\VV^{\pvar}_{\dvar}$ of size $\D$,
with each entry per data center.
For data center $\dvar$, $\VV^{\pvar}_{\dvar}[\dvar]$
is the timestamp of the last update that happens at $\replica{\pvar}{\dvar}$.
For data center $\ivar \neq \dvar$,
$\VV^{\pvar}_{\dvar}[\ivar]$ is the largest timestamp
of the updates that happened at replica $\replica{\pvar}{\ivar}$
and have been propagated to $\replica{\pvar}{\dvar}$.
For fault-tolerance, we compute the \emph{local stable time} $\lst^{\pvar}_{\dvar}$
at replica $\replica{\pvar}{\dvar}$
as the \emph{$(f+1)$-st minimum} element of its $\VV^{\pvar}_{\dvar}$.

Each replica $\replica{\pvar}{\dvar}$ also maintains a
\emph{lst vector} $\LV^{\pvar}_{\dvar}$ of size $\P$,
with each entry per partition.
For partition $1 \le \jvar \le \P$,
$\LV^{\pvar}_{\dvar}[\jvar]$ is the largest $\lstvar$
of replica $\replica{\jvar}{\dvar}$
of which $\replica{\pvar}{\dvar}$ is aware.
We compute the \emph{global stable time} $\gst^{\pvar}_{\dvar}$
at replica $\replica{\pvar}{\dvar}$
as the minimum element of its $\LV^{\pvar}_{\dvar}$.
That is, $\gst^{\pvar}_{\dvar} \triangleq
\min_{1 \le \jvar \le P} \LV^{\pvar}_{\dvar}[\jvar]$.

Each replica $\replica{\pvar}{\dvar}$
periodically synchronizes their $\gstvar^{\pvar}_{\dvar}$
with others via PBFT,
and maintains a \emph{common global stable time} $\cgst^{\pvar}_{\dvar}$.
We discuss it in Section~\ref{ss:cgst}.
%%%%%%%%%%%%%%%%%%%%
\subsubsection{Propagation} \label{sss:propagation}

%%%%%%%%%%%%%%%%%%%%
% metadata.tex

\begin{algorithm}[t]
  \caption{Updating metadata at replica $\replica{\pvar}{\dvar}$}
  \label{alg:metadata}
  \begin{algorithmic}[1]
	\WhenReceive[$\sign{\Call{\replicate}{\vvvar}}{\pvar}{i}$]
	  \label{line:procedure-replicate}
	  \State $\store \gets \store \cup \set{\vvvar}$
	    \label{line:replicate-store}
      \State $\VV[i] \gets \max\set{\VV[i], \vvvar.\utvar}$
	    \label{line:replicate-vv}
	\EndWhenReceive

	\Statex
	\Procedure{\broadcast}{\null}
	  \Comment{Run periodically}
	  \label{line:procedure-broadcast}
      \State \underline{$\lst \gets \max\set{\lst, \text{the $(f+1)$-st minimum element of } \VV[i]}$}
	    \label{line:broadcast-lst}
	  \State $\send \sign{\Call{\bc}{\lst}}{\pvar}{\dvar}
	    \sendto \datacenter(\dvar)$
		\label{line:broadcast-send-bc}
	\EndProcedure

	\Statex
  	\WhenReceive[$\sign{\Call{\bc}{\lstvar}}{j}{\dvar}$]
	  \label{line:procedure-bc}
      \State $\LV[j] \gets \max\set{\LV[j], \lstvar}$
	    \label{line:bc-lv}
      \State \underline{$\gst \gets \max\set{\gst, \min\limits_{1 \le j \le \P} \LV[j]}$}
	    \label{line:bc-gst}
      \If{$\gst > \cgst$}
	    \label{line:bc-gst-if}
      	\State \underline{$\send \sign{\Call{\newcgst}{\gst}}{\pvar}{\dvar}
		  \sendto \servers$}
		  \label{line:bc-send-newcgst}
	  \EndIf
    \EndWhenReceive

	\Statex
	\Procedure{\heartbeat}{\null}
	  \Comment{Run periodically}
	  \label{line:procedure-heartbeat}
	  \State \pre $\clock \ge \VV[\dvar] + \Delta$
	    \label{line:heartbeat-pre}

	  \hStatex
	  \State $\send \sign{\Call{\hb}{\clock}}{\pvar}{\dvar} \sendto \replicas(p)$
		\label{line:heartbeat-send-hb}
	\EndProcedure

	\Statex
	\WhenReceive[$\sign{\Call{\hb}{\clockvar}}{\pvar}{\ivar}$]
	  \label{line:procedure-hb}
	  \State $\VV[\ivar] \gets \max\set{\VV[\ivar], \clockvar}$
	    \label{line:hb-vv}
	\EndWhenReceive
  \end{algorithmic}
\end{algorithm}
%%%%%%%%%%%%%%%%%%%%

As in \gentlerain,
\byzgentlerain{} propagates and updates metadata in the background.
Once a new version $\vvvar$ is created at replica $\replica{\pvar}{\dvar}$,
the replica sends a signed $\replicate(\vvvar)$ message
to other replicas of partition $\pvar$
(line~\code{\ref{alg:replica}}{\ref{line:putreq-send-replicate}}).

When replica $\replica{\pvar}{\dvar}$ receives a $\replicate(\vvvar)$ message
from another replica $\replica{\pvar}{i}$ in data center $i \neq d$,
it stores $\vvvar$ in its $\store^{\pvar}_{\dvar}$
(line~\code{\ref{alg:metadata}}{\ref{line:replicate-store}}),
and updates $\VV^{\pvar}_{\dvar}[\ivar]$ to $\vvvar.\utvar$
if the latter is larger
(line~\code{\ref{alg:metadata}}{\ref{line:replicate-vv}}).

Each replica $\replica{\pvar}{\dvar}$ periodically
computes its $\lst^{\pvar}_{\dvar}$
(line~\code{\ref{alg:metadata}}{\ref{line:broadcast-lst}})
and sends a signed $\broadcast(\lst^{\pvar}_{\dvar})$ message
to $\datacenter(\dvar)$, all the servers in data center $\dvar$
(line~\code{\ref{alg:metadata}}{\ref{line:broadcast-send-bc}}).

When replica $\replica{\pvar}{\dvar}$ receives
a $\broadcast(\lstvar)$ message from another replica $\replica{\jvar}{\dvar}$
in data center $d$, it updates $\LV^{\pvar}_{\dvar}$
and $\gst^{\pvar}_{\dvar}$ accordingly
(lines~\code{\ref{alg:metadata}}{\ref{line:bc-lv}}
and~\code{\ref{alg:metadata}}{\ref{line:bc-gst}}).
If the new $\gst^{\pvar}_{\dvar}$
is larger than $\cgst^{\pvar}_{\dvar}$
(line~\code{\ref{alg:metadata}}{\ref{line:bc-gst-if}}),
the replica $\replica{\pvar}{\dvar}$ sends a
signed $\newcgst(\gst^{\pvar}_{\dvar})$ message
to all servers $\servers$ of the key-value store
(line~\code{\ref{alg:metadata}}{\ref{line:bc-send-newcgst}}).

To ensure liveness, a replica $\replica{\pvar}{\dvar}$
periodically (e.g., at time interval $\Delta$;
line~\code{\ref{alg:metadata}}{\ref{line:heartbeat-pre}})
sends a signed $\hb(\clock^{\pvar}_{\dvar})$ heartbeat
to $\replicas(\pvar)$
(line~\code{\ref{alg:metadata}}{\ref{line:heartbeat-send-hb}}).
When replica $\replica{\pvar}{\dvar}$ receives
a heartbeat $\hb(\clockvar)$ message
from replica $\replica{\pvar}{\ivar}$,
it updates its $\VV^{\pvar}_{\dvar}[\ivar]$ to $\clockvar$
if the latter is larger
(line~\code{\ref{alg:metadata}}{\ref{line:hb-vv}}).
%%%%%%%%%%%%%%%%%%%%%%%%%%%%%%
% cgst-op.tex

%%%%%%%%%%%%%%%%%%%%%%%%%%%%%%
\subsection{Synchronization of Global Stable Time} \label{ss:cgst}

%%%%%%%%%%%%%%%%%%%%
% cgst.tex

\begin{algorithm}[p]
  \caption{Updating $\cgst$ at replica $\replica{\pvar}{\dvar}$
    (see Table~\ref{table:predicates} in Appendix~\ref{section:proof}
	  for the definitions of $\ValidNewLeader$ and $\safepropose$
	  that are adapted from~\cite{ByzLive:DISC20}.)}
  \label{alg:cgst}
  \begin{algorithmic}[1]
	\Statex
	$\begin{aligned}
	  &\safecollect(\storevar) \triangleq
	    \forall \Call{\collectack}{\_, \stvar_{\ivar}} \in \storevar,
	    \forall \sign{\_, \_, \_, \_}{}{\cvar} = \uvar \in \stvar_{\ivar}.\; \valid(\uvar)
	 \end{aligned}$

	\WhenReceive[\underline{$\sign{\Call{\newcgst}{\gstvar}}{\jvar}{\ivar} = \mvar$}]
	  \label{line:procedure-newcgst}
	  \State \pre $\gstvar \le \lst \land \gstvar > \currcgst$
	    \label{line:newcgst-pre}

	  \hStatex
	  \State $\currcgst \gets \gstvar$
	    \label{line:newcgst-currcgst}
	  \State $\Call{\newview}{\viewvar}$ {\bf with} $\viewvar > \currview$
		\label{line:newcgst-call-newview}
	\EndWhenReceive

	\Statex
	\Upon[$\Call{\newview}{\viewvar}$]
	  \label{line:procedure-newview}
	  \State \pre $\viewvar > \currview$
	    \label{line:newview-pre}

	  \hStatex
	  \State $\currview \gets \viewvar$
	    \label{line:newview-curview}
	  \State $\voted \gets \false$
	    \label{line:newview-voted}
	  \State $\send \sign{\Call{\newleader}{\currview, \preparedview,
		\underline{\currcgst, \preparedstore}, \cert}}{\pvar}{\dvar}$ \\
		\hspace{16pt} $\sendto \leader(\currview)$
	    \label{line:newview-send-newleader}
	\EndUpon

	\Statex
	\WhenReceive[$\set{\sign{\Call{\newleader}{\viewvar, \viewvar_{\ivar},
	  \underline{\cgstvar_{\ivar}, \storevar_{\ivar}}, \certvar_{\ivar}}} {\pvar}{\ivar}
	  \mid \replica{\pvar}{\ivar} \in Q} = \M$ \\
	  \hspace{6pt} {\bf from a quorum $Q$}]
	    \label{line:procedure-newleader}
      \State \pre $\currview = \viewvar \land
	  			   \leader(\viewvar) = \replica{\pvar}{\dvar} \land
				   (\forall \mvar \in \M.\; \ValidNewLeader(\mvar))$
		\label{line:newleader-pre}

	  \hStatex
	  \If{$\exists \jvar.\; \viewvar_{\jvar} = \max\set{\viewvar_{\ivar}
	    \mid \replica{\pvar}{\ivar} \in Q} \neq 0$}
		\label{line:newleader-if}
	    \State	$\send \sign{\Call{\propose}{\viewvar,
		  \underline{\storevar_{\jvar}}, \M}}
		  {\pvar}{\dvar} \sendto \replicas(\pvar)$
		  \label{line:newleader-send-propose-if}
	  \Else
	    \State \underline{$\currcgst \gets \max\set{\cgstvar_{\ivar}
		  \mid \replica{\pvar}{\ivar} \in Q \land \cgstvar_{\ivar} \le \lst^{\pvar}_{\dvar}}$}
		  \label{line:newleader-currcgst}
	  	\State \underline{$\send \sign{\Call{\collect}{\currcgst}}{\pvar}{\dvar}
		  \sendto \Call{\replicas}{\pvar}$}
		  \label{line:newleader-send-collect}
		\State \underline{\wait\receive $\set{\sign{\Call{\collectack}{
		    \currcgst, \stvar_{\ivar}}}{\pvar}{\ivar}
		    \mid \replica{\pvar}{\ivar} \in Q'} = \storevar$} \\
		  \hspace{35pt} \underline{{\bf from a quorum $Q'$} {\bf satisfying}
		  $\safecollect(\storevar)$}
		  \label{line:newleader-wait-receive-collectack}

		\hStatex
	  	\State $\send \sign{\Call{\propose}{\viewvar, \underline{\storevar},
		  \M}}{\pvar}{\dvar} \sendto \replicas(\pvar)$
		  \label{line:newleader-send-propose-else}
	  \EndIf
	\EndWhenReceive

	\Statex
	\WhenReceive[$\sign{\Call{\propose}{\viewvar, \storevar, \M}}
	  {\pvar}{\ivar} = \mvar$]
	  \label{line:procedure-propose}
	  \State \pre $\currview = \viewvar \land
	               \voted = \false \land
	  			   \safepropose(\mvar) \land
				   \underline{\safecollect(\storevar)}$
	    \label{line:propose-pre}

	  \hStatex
	  \State \underline{$\currstore \gets \storevar$}
	    \label{line:propose-currval}
	  \State $\voted \gets \true$
	    \label{line:propose-voted}
      \State $\send \sign{\Call{\prepared}{\viewvar,
	    \underline{\hash(\currstore)}}}
	    {\pvar}{\dvar} \sendto \replicas(\pvar)$
		\label{line:propose-send-prepared}
  	\EndWhenReceive

	\Statex
  	\WhenReceive[$\set{\sign{\Call{\prepared}{\viewvar, \hvar}}{\pvar}{\ivar}
	  \mid \replica{\pvar}{\ivar} \in Q} = \C$ {\bf from a quorum $Q$}]
	  \label{line:procedure-prepared}
	  \State \pre $\currview = \viewvar \land
	               \voted = \true \land
	               \underline{\hash(\currstore) = \hvar}$
	    \label{line:prepared-pre}

	  \hStatex
	  \State $\preparedview \gets \currview$
	    \label{line:prepared-preparedview}
	  \State $\preparedstore \gets \currstore$
	    \label{line:prepared-preparedstore}
	  \State $\cert \gets \C$
	    \label{line:prepared-cert}
	  \State $\send \sign{\Call{\commit}{\viewvar, \hvar}}{\pvar}{\dvar}
	    \sendto \replicas(\pvar)$
		\label{line:prepared-send-commit}
  	\EndWhenReceive

	\Statex
	\WhenReceive[$\set{\sign{\Call{\commit}{\viewvar, \hvar}}{\pvar}{\ivar}
	  \mid \replica{\pvar}{\ivar} \in Q}$ {\bf from a quorum $Q$}]
	  \label{line:procedure-commit}
	  \State \pre $\currview = \preparedview = \viewvar \land
	               \underline{\hash(\currstore) = \hvar}$
	    \label{line:commit-pre}

	  \hStatex
	  \State \underline{$\storevar \gets \bigcup\big\{\stvar_{\ivar} \mid
		\sign{\Call{\collectack}{\cgstvar, \stvar_{\ivar}}}{\pvar}{\ivar}
	    \in \currstore\big\}$}
		\label{line:commit-storevar}
	  \If{$\cgst < \cgstvar$}
	    \label{line:commit-cgst-if}
		\State \underline{$\cgst \gets \cgstvar$}
		  \label{line:commit-cgst}
		\State \underline{$\store \gets \storevar \cup
		  \set{\sign{\_, \_, \utvar, \cvar}{}{\cvar} \in \store
	      \mid \utvar > \currcgst}$}
	      \label{line:commit-store}
	  \EndIf
	\EndWhenReceive

	\Statex
	\WhenReceive[\underline{$\sign{\Call{\collect}{\cgstvar}}{\pvar}{\ivar} = \mvar$}]
	  \label{line:procedure-collect}
	  \State \pre $\cgstvar \le \lst \land \cgstvar \ge \currcgst$
	    \label{line:collect-pre}

	  \hStatex
	  \State $\currcgst \gets \cgstvar$
	    \label{line:collect-currcgst}
	  \State $\stvar \gets
	    \set{\sign{\_, \_, \utvar, \cvar}{}{\cvar} \in \store
	    \mid \utvar \le \currcgst}$
		\label{line:collect-stvar}
	  \State $\send \sign{\Call{\collectack}{\cgstvar, \stvar}}{\pvar}{\dvar}
	    \sendto \replica{\pvar}{\ivar}$
		\label{line:collect-send-collectack}
	\EndWhenReceive
  \end{algorithmic}
\end{algorithm}
%%%%%%%%%%%%%%%%%%%%

Each individual partition $\pvar$ independently
runs PBFT~\cite{PBFT:Thesis00} to reach agreement on
a common global stable time $\cgstvar$
\emph{and} the same set of updates before $\cgstvar$
across $\replicas(\pvar)$ (Algorithm~\ref{alg:cgst}).
We follow the pseudocode of single-shot PBFT
described in~\cite{ByzLive:DISC20},
and refer its detailed description
and correctness proof to~\cite{ByzLive:DISC20}.
In the following, we elaborate the parts
specific to synchronization of global stable time;
see the pseudocode underlined in Algorithm~\ref{alg:cgst}.

When replica $\replica{\pvar}{\dvar}$ receives
a $\newcgst(\gstvar)$ message,
it first checks whether $\gstvar \le \lst^{\pvar}_{\dvar}$ as expected
and $\gstvar > \currcgst^{\pvar}_{\dvar}$
which means that no smaller global stable time has been
or is being synchronized
(line~\code{\ref{alg:cgst}}{\ref{line:newcgst-pre}}).
If so, it sets $\currcgst$ to $\gstvar$
(line~\code{\ref{alg:cgst}}{\ref{line:newcgst-currcgst}}).
Now the replica stops accepting updates
with timestamps $< \currcgst^{\pvar}_{\dvar}$
(line~\code{\ref{alg:replica}}{\ref{line:putreq-pre}}).
Then it triggers a $\newview$ action
with a $\viewvar$ larger than $\currview^{\pvar}_{\dvar}$
(line~\code{\ref{alg:cgst}}{\ref{line:newcgst-call-newview}}).

As in classic PBFT~\cite{PBFT:Thesis00,ByzLive:DISC20},
the $\newview(\viewvar)$ action can also be triggered spontaneously,
due to timeout, or by failure detectors.
When it is triggered at a replica $\replica{\pvar}{\dvar}$,
the replica will send a signed $\newleader$ message
to the leader $\leader(\viewvar)$ of $\viewvar$ in $\replicas(\pvar)$
(line~\code{\ref{alg:cgst}}{\ref{line:newview-send-newleader}}).
The $\newleader$ message carries both $\currcgst^{\pvar}_{\dvar}$
and $\preparedstore^{\pvar}_{\dvar}$
which is the set of updates collected in $\preparedview^{\pvar}_{\dvar}$.

When replica $\replica{\pvar}{\dvar}$ receives
a set $\M$ of $\newleader$ messages
from a quorum $Q$ of $\replicas(\pvar)$,
it selects as its proposal from $\M$
the set $\storevar_{\jvar}$ of collected updates
that is prepared in the highest view, say $\viewvar_{\ivar}$
(line~\code{\ref{alg:cgst}}{\ref{line:newleader-send-propose-if}}),
or, if there are no such $\storevar_{\jvar}$, its own proposal.
In the latter case, the replica sets its $\currcgst^{\pvar}_{\dvar}$
to the maximum of $\cgstvar_{\ivar}$ in $Q$
that are $\le \lst^{\pvar}_{\dvar}$
(line~\code{\ref{alg:cgst}}{\ref{line:newleader-currcgst}}).
Then, it sends a signed $\collect(\currcgst^{\pvar}_{\dvar})$
message to $\replicas(\pvar)$
(line~\code{\ref{alg:metadata}}{\ref{line:newleader-send-collect}}),
and waits to receive enough $\collectack$ messages.

When replica $\replica{\pvar}{\dvar}$ receives
a $\collect(\cgstvar)$ message from replica $\replica{\pvar}{\ivar}$
and $\cgstvar$ passes the precondition
(line~\code{\ref{alg:cgst}}{\ref{line:collect-pre}}),
it first sets its $\currcgst^{\pvar}_{\dvar}$ to $\cgstvar$
(line~\code{\ref{alg:cgst}}{\ref{line:collect-currcgst}}).
Now the replica stops accepting updates
with timestamps $< \currcgst^{\pvar}_{\dvar}$
(line~\code{\ref{alg:replica}}{\ref{line:putreq-pre}}).
Then it sends a signed $\collectack(\cgstvar, \stvar)$ message
back to $\replica{\pvar}{\ivar}$
(line~\code{\ref{alg:cgst}}{\ref{line:collect-send-collectack}}),
where $\stvar$ is the set of updates
with timestamps $\le \currcgst^{\pvar}_{\dvar}$
in its $\store^{\pvar}_{\dvar}$
(line~\code{\ref{alg:cgst}}{\ref{line:collect-stvar}}).

The replica $\replica{\pvar}{\dvar}$ waits to receive
a set, denoted $\storevar$, of $\collectack$ messages
from a quorum $Q'$ of $\replicas(\pvar)$.
We require the messages in $\storevar$ carry the same
$\currcgst$ as in the corresponding $\collect$ message
and the signatures of all the collected updates be valid
(i.e., $\safecollect(\storevar)$ holds).
Then, it sends a signed $\propose$ message
with $\storevar$ as its proposal to $\replicas(\pvar)$
(line~\code{\ref{alg:cgst}}{\ref{line:newleader-send-propose-else}}).

When replica $\replica{\pvar}{\dvar}$ receives
a $\propose$ message from replica $\replica{\pvar}{\ivar}$,
it also checks the predicate $\safecollect(\storevar)$
(line~\code{\ref{alg:cgst}}{\ref{line:propose-pre}}).
After setting $\currcgst^{\pvar}_{\dvar}$ to $\storevar$,
it sends a signed $\prepared$ message to $\replicas(\pvar)$.
When replica $\replica{\pvar}{\dvar}$ receives
a set $\C$ of $\prepared$ messages
from a quorum $Q$ of $\replicas(\pvar)$,
both its $\currview^{\pvar}_{\dvar}$
and $\currstore^{\pvar}_{\dvar}$ are prepared
(lines~\code{\ref{alg:cgst}}{\ref{line:prepared-preparedview}}
and~\code{\ref{alg:cgst}}{\ref{line:prepared-preparedstore}}).
The certification $\C$ is also remembered in $\cert^{\pvar}_{\dvar}$
(line~\code{\ref{alg:cgst}}{\ref{line:prepared-cert}}).
They will be sent to new leaders in view changes
to ensure agreement across views
(line~\code{\ref{alg:cgst}}{\ref{line:newview-send-newleader}}).
Then the replica sends a signed $\commit$ message to $\replicas(\pvar)$
(line~\code{\ref{alg:cgst}}{\ref{line:prepared-send-commit}}).

When replica $\replica{\pvar}{\dvar}$ receives
a set of $\commit$ message from a quorum $Q$ of $\replicas(\pvar)$,
it computes $\storevar$ as the union of
the sets of updates $\stvar_{i}$ collected
from each $\replica{\pvar}{\ivar}$ in $\currstore^{\pvar}_{\dvar}$
(line~\code{\ref{alg:cgst}}{\ref{line:commit-storevar}}).
If $\cgst^{\pvar}_{\dvar}$ is smaller than the $\cgstvar$
in $\currstore^{\pvar}_{\dvar}$,
the replica sets $\cgst^{\pvar}_{\dvar}$ to this $\cgstvar$
(line~\code{\ref{alg:cgst}}{\ref{line:commit-cgst-if}}),
and replaces the set of updates with timestamps
$\le \currcgst^{\pvar}_{\dvar}$ in $\store^{\pvar}_{\dvar}$
with the new $\storevar$
(line~\code{\ref{alg:cgst}}{\ref{line:commit-store}}).
%%%%%%%%%%%%%%%%%%%%%%%%%%%%%%
% \input{sections/discussion}
%%%%%%%%%%%%%%%%%%%%%%%%%%%%%%
%%%%%%%%%%%%%%%%%%%%%%%%%%%%%%%%%%%%%%%%
% evaluation.tex

%%%%%%%%%%%%%%%%%%%%%%%%%%%%%%%%%%%%%%%%
\section{Evaluation} \label{section:evaluation}

We evaluate \byzgentlerain{} in terms of
performance, throughput, and latency of remote update visibility.
We also compare \byzgentlerain{} to \byzrcm.
%%%%%%%%%%%%%%%%%%%%%%%%%%%%%%
\subsection{Implementation and Setup} \label{ss:impl-setup}

We implement both \byzgentlerain{} and \byzrcm{} in Java
and use Google's Protocol Buffers~\footnote{
  Protocol Buffers: \url{https://developers.google.com/protocol-buffers}.}
for message serialization.
We implement the key-value stores as \textsf{HashMap},
where each key is associated with a linked list of versioned values.
The key-value stores hold 300 keys in main memory,
with each key of size 8 bytes and each value of size 64 bytes.

We run all experiments on 4 Aliyun~\footnote{
  Alibaba Cloud: \url{https://www.alibabacloud.com/}.} instances running Ubuntu 16.04.
Each instance is configured as a data center,
with 1 virtual CPU core, 300 MB memory, and 1G SSD storage.
All keys are shared into 3 partitions within each data center,
according to their hash values.
%%%%%%%%%%%%%%%%%%%%%%%%%%%%%%
\subsection{Evaluation Results} \label{ss:results}

%%%%%%%%%%%%%%%%%%%%
\begin{figure}[t]
  \centering
  \begin{subfigure}[c]{0.45\textwidth}
    \centering
    \includegraphics[width = \textwidth]{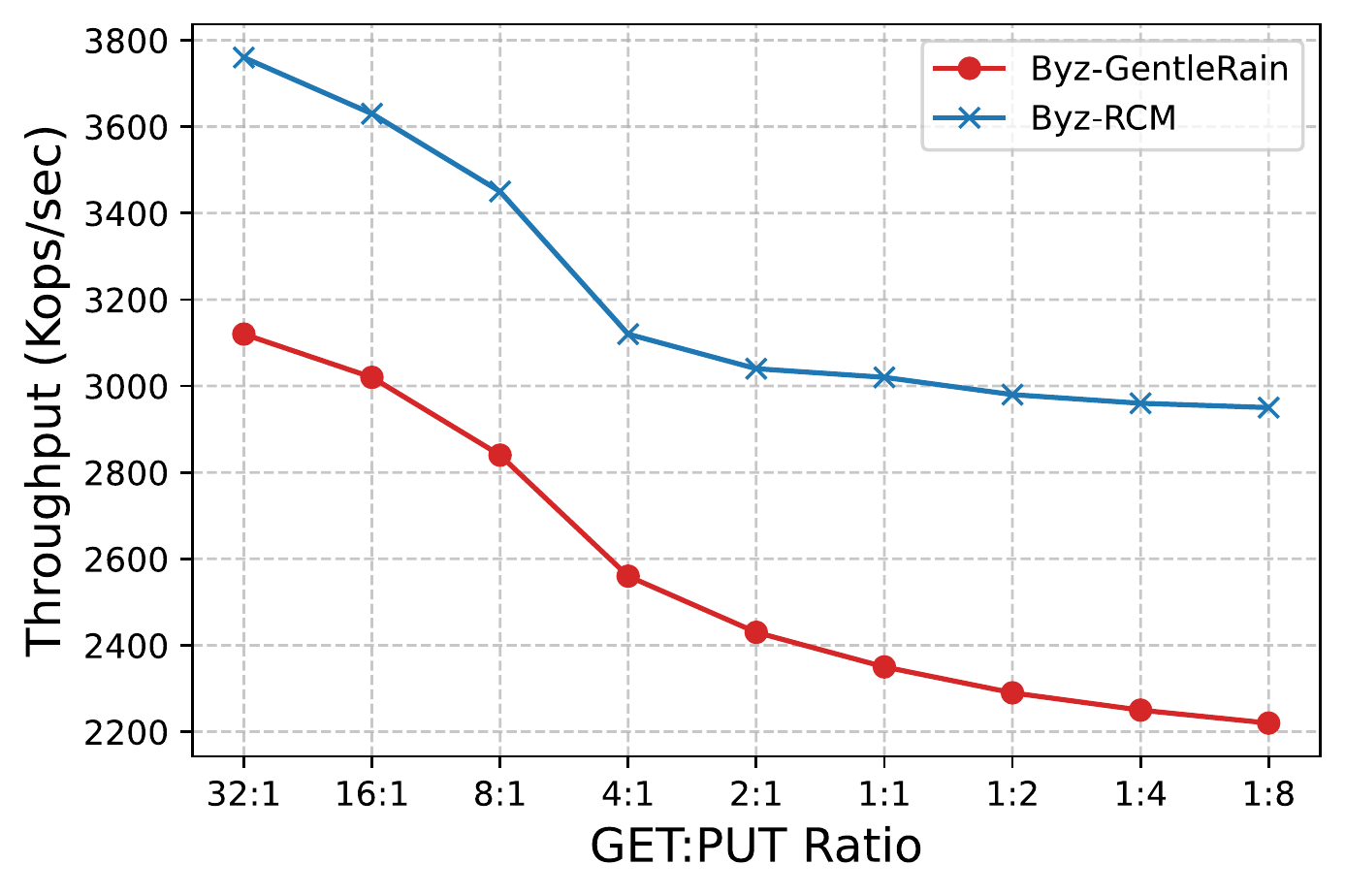}
    \caption{Throughput}
    \label{fig:ff-throughput}
  \end{subfigure}
  \hfill
  \begin{subfigure}[c]{0.45\textwidth}
    \centering
    \includegraphics[width = \textwidth]{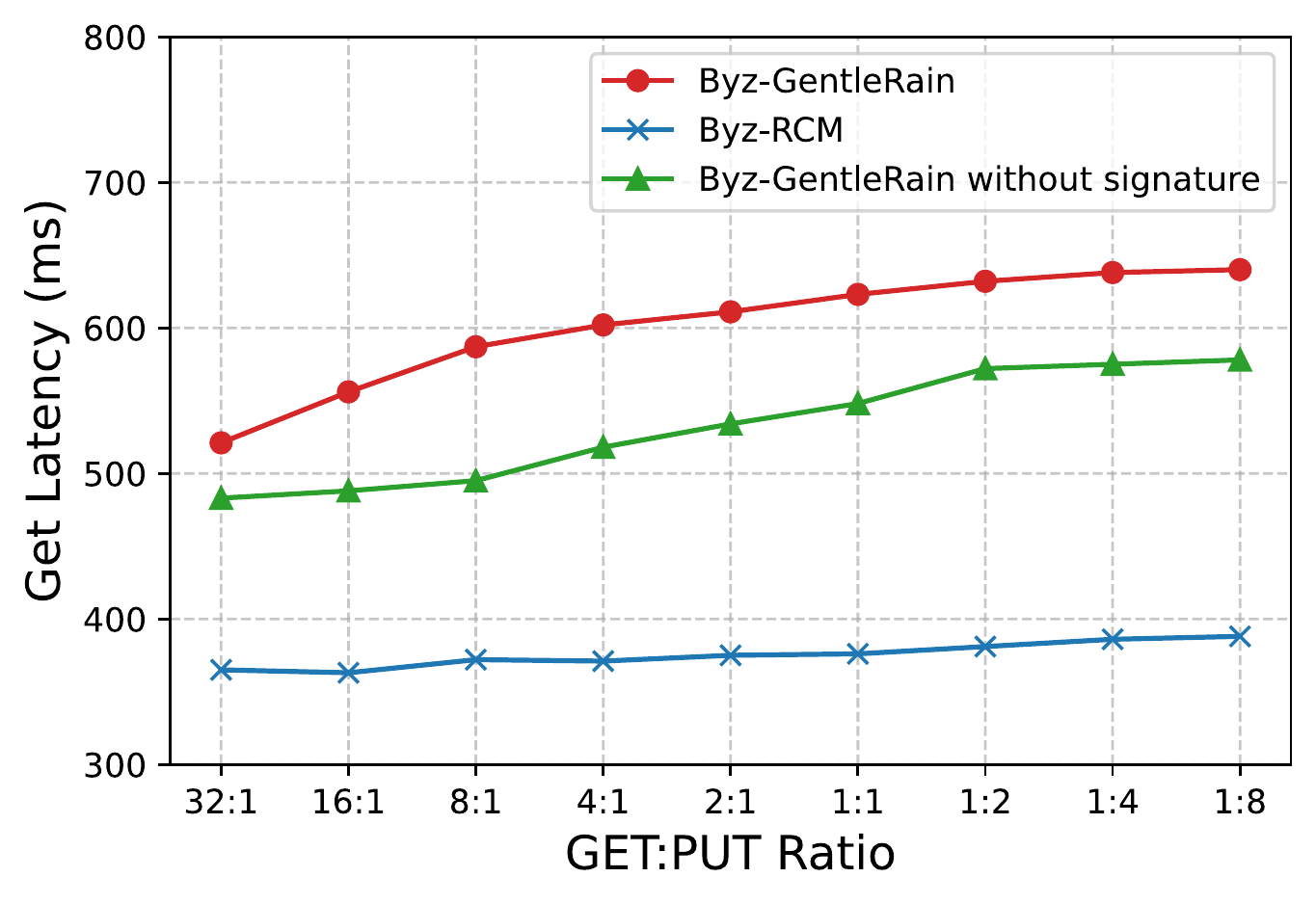}
    \caption{\get{} latency}
    \label{fig:ff-get-latency}
  \end{subfigure}
  \hfill
  \begin{subfigure}[c]{0.45\textwidth}
    \centering
    \includegraphics[width = \textwidth]{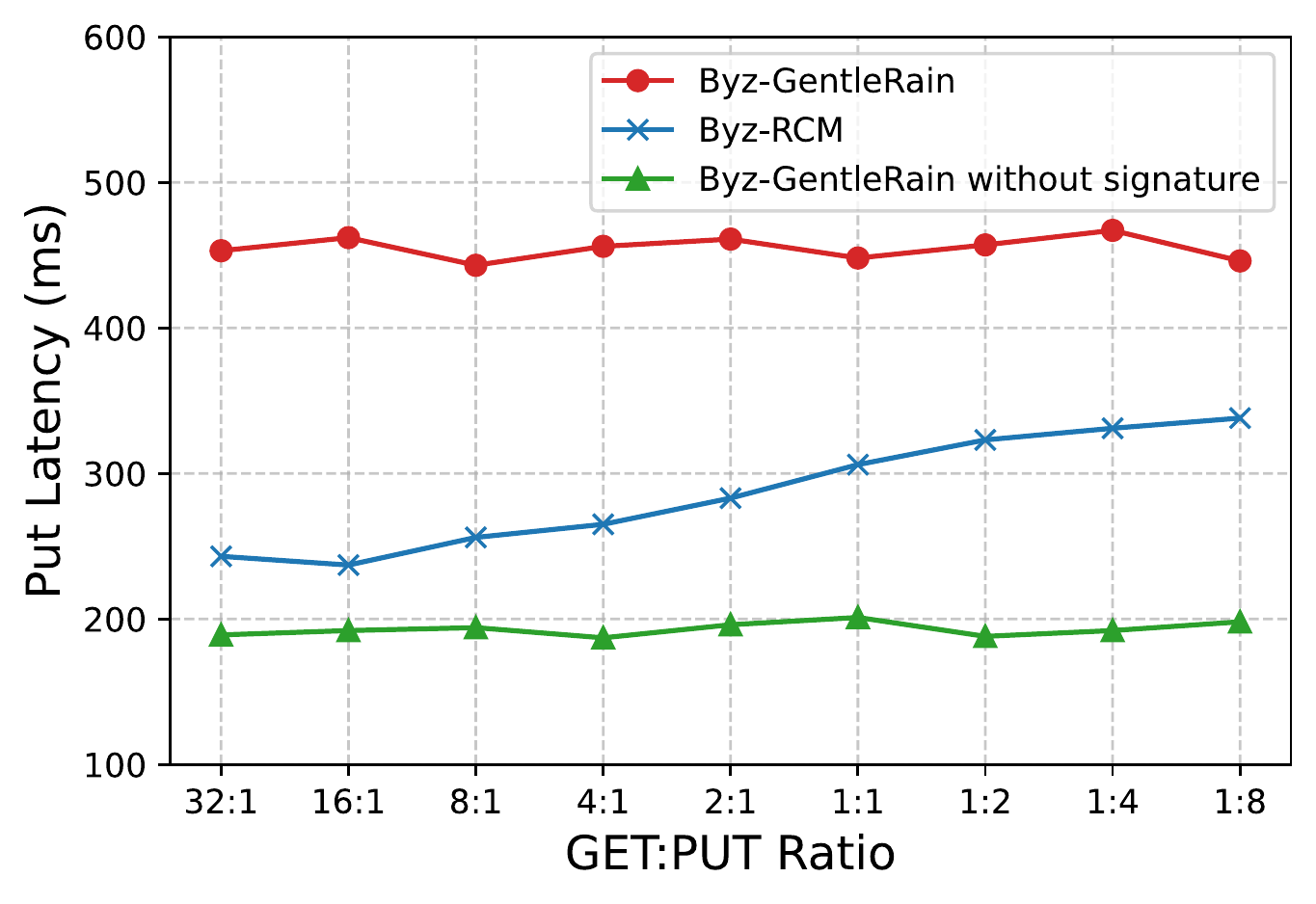}
    \caption{\put{} latency}
    \label{fig:ff-put-latency}
  \end{subfigure}
  \hfill
  \begin{subfigure}[c]{0.45\textwidth}
    \centering
    \includegraphics[width = \textwidth]{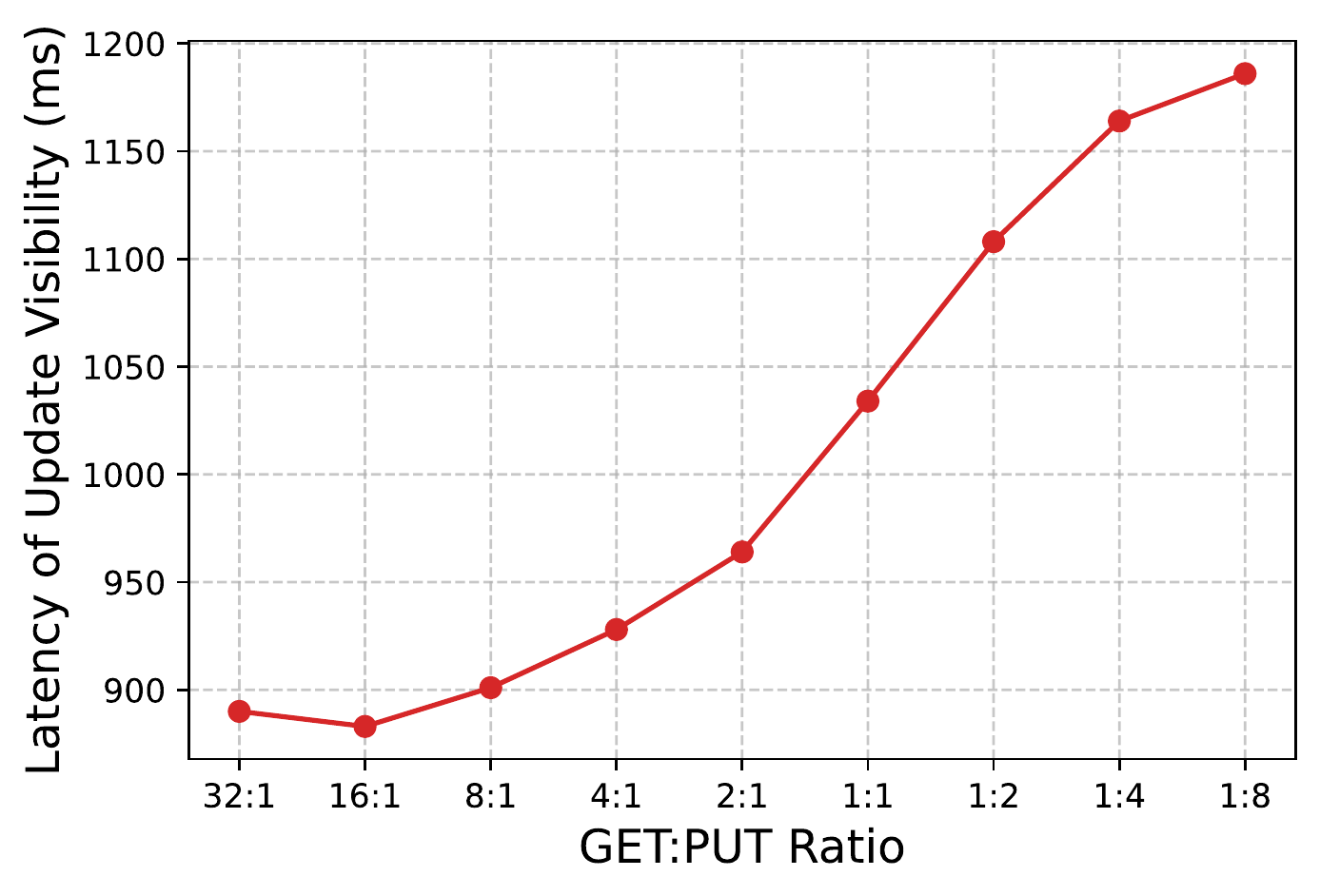}
    \caption{\put{} visibility}
    \label{fig:ff-visibility}
  \end{subfigure}
  \caption{Evaluation of \byzgentlerain{} and \byzrcm{} in failure-free scenarios.}
  \label{fig:ff-comparison}
\end{figure}
%%%%%%%%%%%%%%%%%%%%

Figure~\ref{fig:ff-comparison} shows the system throughput
and the latency of \get{} and \put{} operations
of both \byzgentlerain{} and \byzrcm{} in failure-free scenarios.
We vary the $\get:\put$ ratios of workloads.
First, \byzrcm{} performs better than \byzgentlerain,
especially with low $\get:\put$ ratios.
This is because \byzrcm{} assumes Byzantine fault-free clients
and is \emph{signature-free}.
In contrast, \byzgentlerain{} requires clients sign each $\putreq$ request.
Second, it demonstrates that \byzgentlerain{} is quite efficient
on typical workloads, especially for read-heavy workloads.
Third, the performance of \byzgentlerain{}
is closely comparable to that of \byzrcm,
if digital signatures are omitted deliberately from \byzgentlerain;
see Figures~\ref{fig:ff-get-latency} and \ref{fig:ff-put-latency}.
Finally, Figure~\ref{fig:ff-visibility} shows the latency of \put{} visibility,
which gets higher and higher with more and more \put{} operations.

%%%%%%%%%%%%%%%%%%%%
\begin{figure}[t]
  \begin{subfigure}[c]{0.45\textwidth}
    \centering
    \includegraphics[width = \textwidth]{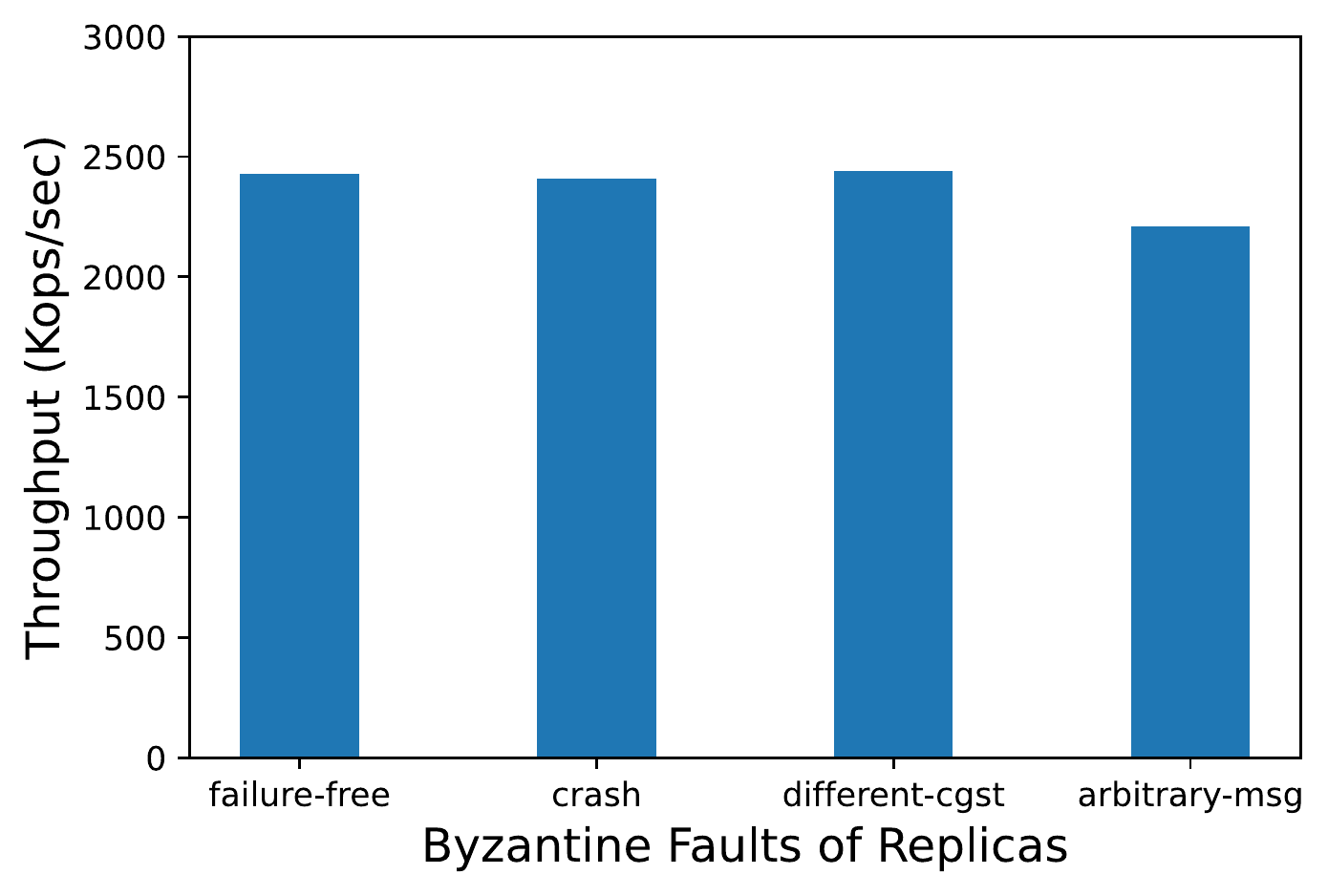}
    \caption{}
    \label{fig:byz-replicas}
  \end{subfigure}
  \hfill
  \begin{subfigure}[c]{0.45\textwidth}
    \centering
    \includegraphics[width = \textwidth]{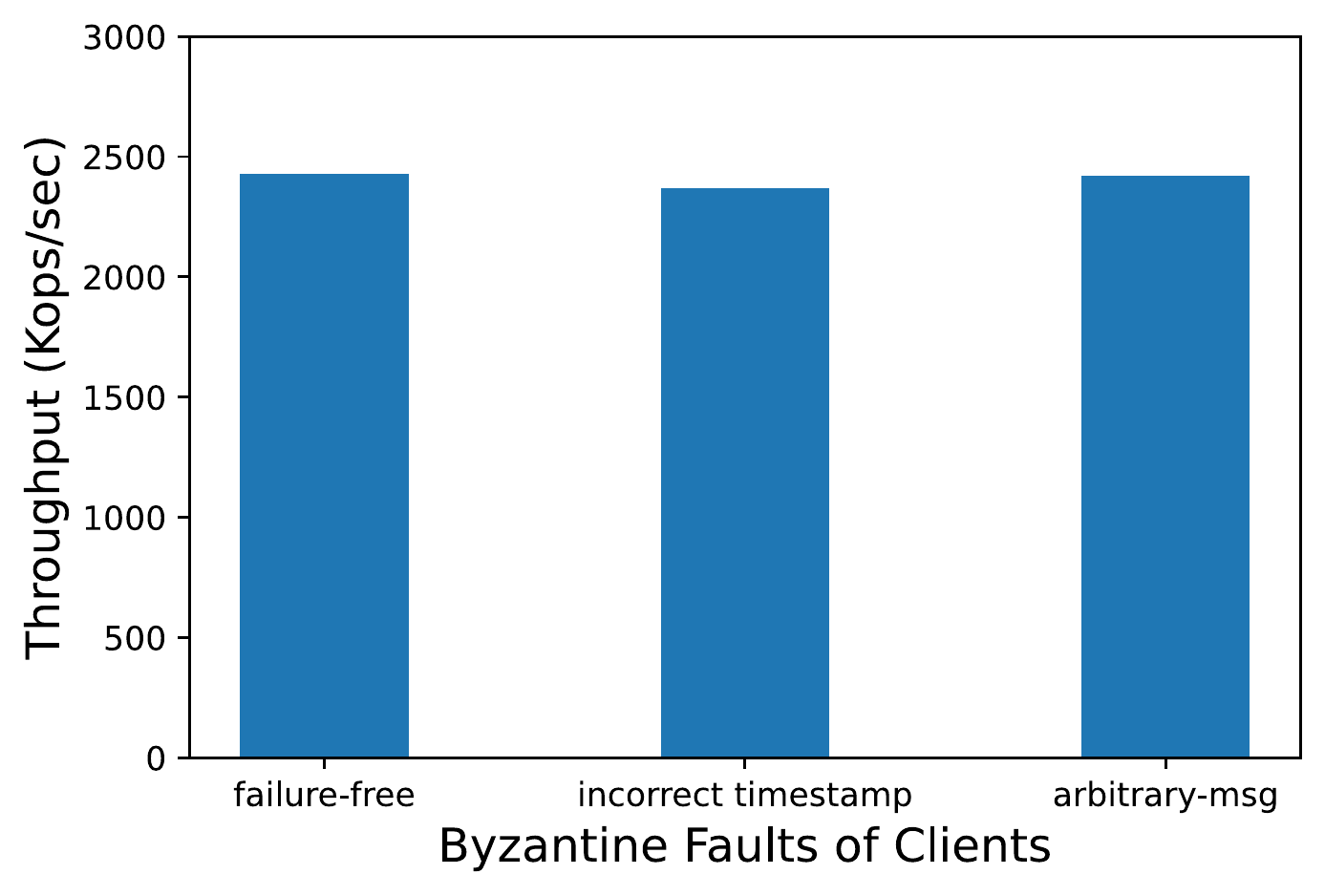}
    \caption{}
    \label{fig:ff-clients}
  \end{subfigure}
  \caption{Evaluation of \byzgentlerain{} in Byzantine scenarios.}
  \label{fig:byz}
\end{figure}
%%%%%%%%%%%%%%%%%%%%

We also evaluate \byzgentlerain{} in several typical Byzantine scenarios.
Generally, both Byzantine clients and replicas
may fail by crash or send arbitrary messages.
Particularly, we consider
\begin{inparaenum}[(1)]
  \item Byzantine clients that may send \getreq{} and/or \putreq{}
    requests with incorrect timestamps
    (line~\code{\ref{alg:client}}{\ref{line:get-send-getreq}}
    and line~\code{\ref{alg:client}}{\ref{line:put-send-putreq}}), and
  \item Byzantine replicas that may broadcast
    different global stable time $\cgstvar$ to replicas in different partitions
    (line~\code{\ref{alg:metadata}}{\ref{line:bc-send-newcgst}}).
\end{inparaenum}

Figure~\ref{fig:byz} demonstrates the impacts of
various Byzantine failures on the system throughput of \byzgentlerain.
On the one hand, the Byzantine failures of types (1) and (2) above
has little impact on throughput.
On the other hand, frequently sending arbitrary messages,
such as \newcgst{} or \propose{} messages, does hurt throughput.
This is probably due to the signatures carried by these messages.
%%%%%%%%%%%%%%%%%%%%%%%%%%%%%%%%%%%%%%%%
% related-work.tex

%%%%%%%%%%%%%%%%%%%%%%%%%%%%%%%%%%%%%%%%
\section{Related Work} \label{section:related-work}

As far as we know, Byz-RCM~\cite{ByzRCM:NAC19}
is the only causal consistency protocol
that considers Byzantine faults.
It achieves causal consistency in the client-server model with $3f + 1$ servers
where up to $f$ servers may suffer Byzantine faults,
and any number of clients may crash.
Byz-RCM has also been shown optimal in terms of failure resilience.
However, Byz-RCM did not tolerate \emph{Byzantine clients},
and thus it could rely on clients' requests
to identify bogus requests from Byzantine servers~\cite{ByzRCM:NAC19}.

Linde \emph{et. al.}~\cite{ClientSide:VLDB20}
consider \emph{peer-to-peer} architecture.
A centralized server maintains the application data,
while clients replicate a subset of data
and can directly communicate with each other.
They analyze the possible attacks of clients to causal consistency
(the centralized server is assumed correct),
derive a secure form of causal consistency,
and propose practical protocols for implementing it.

Liskov and Rodrigues extend the notion of linearizability~\cite{Lin:TOPLAS90}
and define BFT-linearizability in the presence of
Byzantine servers and clients~\cite{Liskov:ICDCS16}.
They also design protocols that achieve BFT-linearizability
despite Byzantine clients.
The protocols require $3f + 1$ replicas
of which up to $f$ replicas may be Byzantine.
They are quite efficient for linearizable systems:
Writes complete in two or three phases,
while reads complete in one or two phases.

Auvolat \emph{et. al.}~\cite{BTCB:TCS21}
defines a Byzantine-tolerant Causal Order broadcast
(BCO-broadcast) abstraction
and proposes an implementation for it.
However, as a communication primitive for replicas,
BCO-broadcast does not capture the get/put semantics
from the perspective of clients.
Thus, it does not prevent Byzantine clients
from violating causality.
%%%%%%%%%%%%%%%%%%%%%%%%%%%%%%%%%%%%%%%%
% conclusion.tex

%%%%%%%%%%%%%%%%%%%%%%%%%%%%%%%%%%%%%%%%
\section{Conclusion} \label{section:conclusion}

We present \byzgentlerain, the first
causal consistency protocol which tolerates up to $f$ Byzantine
servers among $3f + 1$ servers in each partition
\emph{and} any number of Byzantine clients.
% \byzgentlerain{} combines GentleRain for causal consistency
% and PBFT for dealing with Byzantine faults.
% All reads and updates finish in one round-trip.
The preliminary experiments show that
\byzgentlerain{} is quite efficient on typical workloads.
Yet, more extensive large-scale experiments on more benchmarks are needed.
We will also explore optimizations of our synchronization protocol
in Algorithm~\ref{alg:cgst} in future work.

% Compared to PBFT, our synchronization protocol
% requires exchanging potentially huge \propose{} messages with more signatures.
% In an optimized version,
% the PBFT leader could broadcast the \emph{union} of the sets of updates
% it collects in its \propose{} message.
% When a server receives the \propose{} message,
% it checks if some updates are missing against its local store.
% If so, the server tags its \prepared{} message with a \doubt{} flag;
% otherwise, a \trust{} flag.
% Then each server waits for
% $2f+1$ \prepared{} messages \emph{with} \trust{} flag before proceeding.
% It is possible to further trade failure resilience for performance,
% e.g., requiring $4f+1$ servers in each partition.
% We will explore these optimizations in future work.
%%%%%%%%%%%%%%%%%%%%%%%%%%%%%%%%%%%%%%%%
% ack.tex

\section{Acknowledgements}
This work was partially supported by the CCF-Tencent Open Fund
(CCF-Tencent RAGR20200124) and
the National Natural Science Foundation of China (No. 61772258).
%%%%%%%
\bibliographystyle{splncs04}
\bibliography{byz-gentlerain}
%%%%%%%%%%%%%%%%%%%%%%
\newpage
\setcounter{table}{0}
\renewcommand{\thetable}{A\arabic{table}}
\setcounter{figure}{0}
\renewcommand{\thefigure}{A\arabic{table}}
\begin{appendix}
% cgst-appendix.tex

%%%%%%%%%%%%%%%%%%%%
% predicates.tex

\begin{table}[t]
  \centering
  \caption{Predicates in Algorithm~\ref{alg:cgst}
    (adapted from~\cite{ByzLive:DISC20}).}
  \label{table:predicates}
  \begin{tabular}{c}
    \hline
	$\begin{aligned}
	    &\preparedpred(\viewvar, \hvar, \M) \triangleq \\
		  &\quad \exists Q.\; \quorum(Q) \land
	        \M = \set{\sign{\Call{\prepared}{\viewvar, \hvar, \M}}{\pvar}{\ivar}
	          \mid \replica{\pvar}{\ivar} \in Q}
	 \end{aligned}$ \\
	\hline
	$\begin{aligned}
	  &\ValidNewLeader(\sign{\Call{\newleader}{\viewvar, \viewvar_{\ivar}, \cgstvar_{\ivar}, \certvar_{\ivar}}}
	    {\pvar}{\ivar}) \triangleq \\
		&\quad \viewvar_{\ivar} < \viewvar \land (\viewvar_{\ivar} \neq 0 \implies
		  \preparedpred(\viewvar_{\ivar}, \hash(\storevar), \M))
	 \end{aligned}$ \\
	\hline
	$\begin{aligned}
	  &\safepropose(\sign{\Call{\propose}{\viewvar, \storevar, \M}}{\pvar}{\jvar}) \triangleq \\
	  &\qquad \replica{\pvar}{\jvar} = \leader(\viewvar) \\
	  &\quad \land \exists Q, \viewvar, \storevar, \M.\; \quorum(\Q) \\
	  &\quad \land \M = \set{\sign{\Call{\newleader}{\viewvar, \viewvar_{\ivar},
	    \underline{\cgstvar_{\ivar}, \storevar_{\ivar}}, \certvar_{\ivar}}}{\pvar}{\ivar}
	    \mid \replica{\pvar}{\ivar} \in \Q} \\
	  &\quad \land (\forall \mvar \in \M.\; \ValidNewLeader(\mvar)) \\
	  &\quad \land ((\exists \jvar.\; \viewvar_{\jvar} \neq 0)
	    \implies (\exists \jvar.\; \viewvar_{\jvar}
		  = \max\set{\viewvar_{\ivar} \mid \replica{\pvar}{\ivar} \in \Q}
		    \land \storevar = \storevar_{\jvar}))
	 \end{aligned}$ \\
	\hline
  \end{tabular}
\end{table}
%%%%%%%%%%%%%%%%%%%%
% proof.tex

%%%%%%%%%%%%%%%%%%%%%%%%%%%%%%%%%%%%%%%%
\section{Correctness of \byzgentlerain} \label{section:proof}

We show that \byzgentlerain{} satisfies \byzcc.
We assume that single-shot PBFT is correct
and refer its detailed correctness proof to~\cite{ByzLive:DISC20}.
Table~\ref{table:predicates} gives the definitions
of the predicates $\ValidNewLeader$ and $\safepropose$
used in Algorithm~\ref{alg:cgst},
which are also adapted from~\cite{ByzLive:DISC20}.

\begin{remark} \label{remark:notations}
  In the following, we use $\R$ and $\W$ to denote the set of
  \get{} and \put{} operations, respectively.
  We also define $\O \triangleq \R \cup \W$
  to denote the set of all operations.

  For a variable, e.g., $\clock_{\cvar}$ at client $\cvar$,
  we refer to its value at time $\timepoint$
  by, e.g., $(\clock_{\cvar})_{\timepoint}$.
\end{remark}

According to the description of Algorithms~\ref{alg:client}
and~\ref{alg:replica},
\begin{lemma} \label{lemma:rules}
  \emph{\byzrule~\ref{rule:get}--\byzrule~\ref{rule:get-ts}}
  are maintained by \emph{\byzgentlerain}.
\end{lemma}

\begin{lemma} \label{lemma:invs}
  \emph{\inv~\ref{inv:cgst-c-put}--\inv~\ref{inv:cgst-updates}}
  are maintained by \emph{\byzgentlerain}.
\end{lemma}
\begin{proof} \label{proof:invs}
  \inv~\ref{inv:cgst-c-put} holds due to
  line~\code{\ref{alg:client}}{\ref{line:put-wait-clock}}.
  \inv~\ref{inv:cgst-replica} holds due to
  the read rule at line~\code{\ref{alg:replica}}{\ref{line:putreq-pre}}.
  By the correctness of single-shot PBFT~\cite{ByzLive:DISC20},
  \inv~\ref{inv:cgst-updates} holds.
\end{proof}

\begin{definition}[Timestamps] \label{def:ts}
  We use $\tsof(\ovar)$ to denote the timestamp of operation $\ovar$,
  which is defined as follows:
  \begin{itemize}
    \item For a \get{} operation $\ovar$,
      $\tsof(\ovar)$ refers to the value of ``$\tsvar$''
      at line~\code{\ref{alg:client}}{\ref{line:get-ts}}.
    \item For a \put{} operation $\ovar$,
      $\tsof(\ovar)$ refers to the value of ``$\utvar$''
      at line~\code{\ref{alg:client}}{\ref{line:put-send-putreq}}.
  \end{itemize}
\end{definition}

\begin{lemma} \label{lemma:so-w-w}
  \[
    (\wvar \rel{\so} \wvar' \land \wvar \in \W \land \wvar' \in \W) \implies
      \tsof(\wvar') > \tsof(\wvar).
  \]
\end{lemma}
\begin{proof} \label{proof:so-w-w}
  Suppose that $\wvar'$ is issued by client $\cvar$
  at time $\timepoint'$.
  \[
    \tsof(\wvar) \le (\dt_{c})_{\timepoint'}
      < (\clock_{\cvar})_{\timepoint'} = \tsof(\wvar').
  \]
\end{proof}

\begin{lemma} \label{lemma:so-r-w}
  \[
    (\rvar \rel{\so} \wvar \land \rvar \in \R \land \wvar \in \W) \implies
      \tsof(\wvar) > \tsof(\rvar).
  \]
\end{lemma}
\begin{proof} \label{proof:so-r-w}
  Suppose that $\rvar$ and $\wvar$ are issued by correct client $\cvar$
  at time $\timepoint_{\rvar}$ and $\timepoint_{\wvar}$, respectively.
  By line~\code{\ref{alg:client}}{\ref{line:get-ts}},
  \[
    \tsof(\rvar) = \max\set{(\cgst_{\cvar})_{\timepoint_{\rvar}},
      (\dt_{\cvar})_{\timepoint_{\rvar}}}.
  \]
  By \inv~\ref{inv:cgst-c-put},
  \[
    \tsof(\wvar) > (\cgst_{\cvar})_{\timepoint_{\wvar}}
                 > (\cgst_{\cvar})_{\timepoint_{\rvar}}.
  \]
  Moreover,
  \[
    \tsof(\wvar) > (\dt_{\cvar})_{\timepoint_{\wvar}}
                 \ge (\dt_{\cvar})_{\timepoint_{\rvar}}.
  \]
  Putting it together yields
  \[
    \tsof(\wvar) > \tsof(\rvar).
  \]
\end{proof}

\begin{lemma} \label{lemma:so-w-r}
  \[
    (\wvar \rel{\so} \rvar \land \wvar \in \W \land \rvar \in \R) \implies
      \tsof(\rvar) \ge \tsof(\wvar).
  \]
\end{lemma}
\begin{proof} \label{proof:so-w-r}
  Suppose that $\rvar$ are issued by correct client $\cvar$ at time $\timepoint$.
  By line~\code{\ref{alg:client}}{\ref{line:get-ts}},
  \[
    \tsof(\rvar) = \max\set{(\cgst_{\cvar})_{\timepoint},
      (\dt_{\cvar})_{\timepoint}}
      \ge (\dt_{\cvar})_{\timepoint}.
  \]
  By line~\code{\ref{alg:client}}{\ref{line:put-dt}},
  \[
    (\dt_{\cvar})_{\timepoint} \ge \tsof(\wvar).
  \]
  Thus,
  \[
    \tsof(\rvar) \ge \tsof(\wvar).
  \]
\end{proof}

\begin{lemma} \label{lemma:so-r-r}
  \[
    (\rvar \rel{\so} \rvar' \land \rvar \in \R \land \rvar' \in \R) \implies
      \tsof(\rvar') \ge \tsof(\rvar).
  \]
\end{lemma}
\begin{proof} \label{proof:so-r-r}
  Suppose that $\rvar$ and $\rvar'$ are issued by correct client $\cvar$
  at time $\timepoint$ and $\timepoint'$, respectively.
  By line~\code{\ref{alg:client}}{\ref{line:get-ts}},
  \[
    \tsof(\rvar) = \max\set{(\cgst_{\cvar})_{\timepoint},
      (\dt_{\cvar})_{\timepoint}},
  \]
  and
  \[
    \tsof(\rvar') = \max\set{(\cgst_{\cvar})_{\timepoint'},
      (\dt_{\cvar})_{\timepoint'}}.
  \]
  Moreover,
  \[
    (\cgst_{\cvar})_{\timepoint'} \ge (\cgst_{\cvar})_{\timepoint}
    \land (\dt_{\cvar})_{\timepoint'} \ge (\dt_{\cvar})_{\timepoint}.
  \]
  Thus,
  \[
    \tsof(\rvar') \ge \tsof(\rvar).
  \]
\end{proof}

\begin{lemma} \label{lemma:so}
  \[
    \ovar \rel{\so} \ovar' \land \ovar \in \O \land \ovar' \in \O
      \implies \tsof(\ovar') \ge \tsof(\ovar).
  \]
\end{lemma}
\begin{proof} \label{proof:so}
  By Lemmas~\ref{lemma:so-w-w}--\ref{lemma:so-r-r}.
\end{proof}

\begin{lemma} \label{lemma:read-from}
  \[
    \wvar \rel{\rf} \rvar \land \wvar \in \W \land \rvar \in \R
      \implies \tsof(\rvar) \ge \tsof(\wvar).
  \]
\end{lemma}
\begin{proof} \label{proof:read-from}
  By the read rule at line~\code{\ref{alg:replica}}{\ref{line:get-val}}.
\end{proof}

\begin{lemma} \label{lemma:hb-ts}
  \[
    \ovar \leadsto \ovar' \land o \in \O \land o' \in \O
      \implies \tsof(\ovar') \ge \tsof(\ovar).
  \]
\end{lemma}
\begin{proof} \label{proof:hb-ts}
  By Lemmas~\ref{lemma:so} and~\ref{lemma:read-from}.
\end{proof}

\begin{lemma} \label{lemma:local-cc}
  Consider $\rvar \in \R$ and $\wvar \in \W$.
  Suppose $\rvar$ reads from some value
  at a correct replica $\replica{\pvar}{\dvar}$
  at time $\timepoint$
  (line~\code{\ref{alg:replica}}{\ref{line:getreq-value}}).
  If $\wvar$ would be added to $\store^{\pvar}_{\dvar}$
  at a later time than $\timepoint$
  (line~\code{\ref{alg:replica}}{\ref{line:putreq-store}})
  then $\tsof(\wvar) > \tsof(\rvar)$.
\end{lemma}
\begin{proof} \label{proof:local-cc}
  By \byzrule~\ref{rule:get},
  \[
    (\cgst^{\pvar}_{\dvar})_{\timepoint} \ge \tsof(\rvar).
  \]
  By \inv~\ref{rule:put},
  \[
    \tsof(\wvar) > (\cgst^{\pvar}_{\dvar})_{\timepoint}.
  \]
  Thus,
  \[
    \tsof(\wvar) > \tsof(\rvar).
  \]
\end{proof}

\begin{lemma} \label{lemma:partition-cc}
  Consider $\rvar \in \R$ and $\wvar \in \W$.
  Suppose the successful $\rvar$ returns at time $\timepoint$
  (line~\code{\ref{alg:replica}}{\ref{line:getreq-value}})
  and the successful
  % (line~\code{\ref{alg:client}}{\ref{line:put-return}})
  $\wvar$ starts at a later time than $\timepoint$ in partition $\pvar$.
  Then $\lnot(\wvar \leadsto \rvar$).
\end{lemma}
\begin{proof} \label{proof:partition-causality}
  By lines~\ref{alg:client}{\ref{line:get-wait-receive-getack}}
  and \ref{alg:client}{\ref{line:put-wait-receive-putack}},
  there is a \emph{correct} replica at which
  $\rvar$ obtains its value
  (line~\code{\ref{alg:replica}}{\ref{line:getreq-value}})
  \emph{before} $\wvar$ is added to the store
  (line~\code{\ref{alg:replica}}{\ref{line:putreq-store}}).
  By Lemma~\ref{lemma:local-cc},
  \[
    \tsof(\wvar) > \tsof(\rvar).
  \]
  By Lemma~\ref{lemma:hb-ts},
  \[
    \lnot(\wvar \leadsto \rvar).
  \]
\end{proof}

\begin{theorem} \label{thm:byzcc}
  \emph{\byzgentlerain} satisfies \byzcc{}.
  That is, when a certain \put{} operation is visible to a client,
  then so are all of its causal dependencies.
\end{theorem}
\begin{proof} \label{proof:byzcc}
  By Lemmas~\ref{lemma:partition-cc}
  and \ref{lemma:cross-partition-cc}.
\end{proof}
%%%%%%%%%%%%%%%%%%%%%

\begin{lemma} \label{lemma:w-successful}
  Suppose a \put{} operation $\wvar$ successfully
  returns in partition $\pvar$ at time $\timepoint$
  (line~\code{\ref{alg:client}}{\ref{line:put-wait-receive-putack}}).
  Then, it will eventually be in $\store^{\pvar}_{\ivar}$
  for each correct data center $\ivar$.
\end{lemma}
\begin{proof} \label{proof:w-successful}
  By line~\code{\ref{alg:client}}{\ref{line:put-wait-receive-putack}}
  and line~\code{\ref{alg:cgst}}{\ref{line:newleader-wait-receive-collectack}},
  there is a \emph{correct} replica in partition $\pvar$ at which
  $\wvar$ is added to the store
  (line~\code{\ref{alg:replica}}{\ref{line:putreq-store}})
  before it is sent to the PBFT leader in the $\collectack$ message
  (line~\code{\ref{alg:cgst}}{\ref{line:collect-send-collectack}}).
  By the correctness of single-shot PBFT~\cite{ByzLive:DISC20},
  it will eventually be in $\store^{\pvar}_{\ivar}$
  for each correct data center $\ivar$.
\end{proof}

\begin{lemma} \label{lemma:update-visible}
  Let $\wvar$ be a successful \put{} operation.
  Then, eventually for each correct replica
  $\replica{\jvar}{\ivar}$ ($1 \le \ivar \le \D, 1 \le \jvar \le \P$),
  $\cgst^{\jvar}_{\ivar} \ge \tsof(\wvar)$.
\end{lemma}

\begin{proof}
  Suppose $\wvar$ successfully returns in partition $\pvar$.
  Then it is added to the stores of
  at least $f+1$ \emph{correct} replicas in $\replicas(\pvar)$.
  By Algorithm~\ref{alg:metadata},
  eventually for each correct replica $\replica{\jvar}{\ivar}$,
  $\lst^{\jvar}_{\ivar} \ge \tsof(\wvar)$ and
  $\gst^{\jvar}_{\ivar} \ge \tsof(\wvar)$.
  By Algorithm~\ref{alg:cgst},
  $\cgst^{\jvar}_{\ivar} \ge \tsof(\wvar)$.
\end{proof}

\begin{lemma}
  Suppose a correct replica sends a $\newcgst(\gstvar)$ message
  (line~\code{\ref{alg:metadata}}{\ref{line:bc-send-newcgst}}).
  Then, there is a $\cgstvar \ge \gstvar$ such that
  eventually for each correct replica $\replica{\jvar}{\ivar}$
  ($1 \le \ivar \le \D, 1 \le \jvar \le \P$),
  $\cgst^{\jvar}_{\ivar} \ge \cgstvar$.
\end{lemma}

\begin{proof}
  A byzantine replica can propose $gst^{\pvar}_{\dvar}$
  to replica $\replica{\pvar}{\dvar}$,
  but it will take effect only when
  $gst^{\pvar}_{\dvar} < lst^{\pvar}_{\dvar}$
  maintained by $\replica{\pvar}{\dvar}$.
  If $gst^{\pvar}_{\dvar}$ is smaller than
  a $gst'$ broadcast by a correct replica, then it will be overwritten at $\replica{\pvar}{\dvar}$.
  And eventually, there must be a $gst$ version $\ge \replica{\pvar}{\dvar}$ broadcast by a correct replica.
  Thus it will eventually by overwritten by a correct $gst$ version.
\end{proof}

\begin{lemma} \label{lemma:cross-partition-cc}
  Consider $\rvar \in \R$ and $\wvar \in \W$.
  Suppose the successful $\rvar$ returns at time $\timepoint$
  (line~\code{\ref{alg:replica}}{\ref{line:getreq-value}})
  and the successful
  $\wvar$ starts at a later time than $\timepoint$
  in any partition $\jvar \neq \pvar$.
  Then $\lnot(\wvar \leadsto \rvar$).
\end{lemma}
\begin{proof} \label{proof:across-partition-cc}
  By Algorithm~\ref{alg:replica}\ref{line:putreq-pre}, $\tsof{\wvar} > \lst$ 
  at each replca who accept $\wvar$. 
  By Algorithm~\ref{alg:replica}\ref{line:getreq-wait-until},
  $\tsof{\rvar} \le \cgst$ at each replica who reply it.
  Since $\lst$ at any correct correct server is an upbound of
  all the $\cgst$ at all correct replicas, $tsof(\wvar) > \lst > \cgst > \tsof(\rvar)$.
  So $\lnot(\wvar \leadsto \rvar)$.
\end{proof}
\end{appendix}
%%%%%%%%%%%%%%%%%%%%%%
\end{document}